\documentclass{article}
\makeatletter
\usepackage{setspace}
\usepackage{rotating}
\usepackage{multirow}
\usepackage{mathtools}
\usepackage{longtable}
\usepackage{enumitem}
\usepackage{tikz}
\usepackage{longtable}
\usetikzlibrary{shapes.geometric, arrows}

\tikzstyle{startstop} = [rectangle, rounded corners, 
text centered, 
draw=black]

\tikzstyle{io} = [trapezium, 
trapezium stretches=true, 
trapezium left angle=70, 
trapezium right angle=110,text centered, 
draw=black,text width=3cm,]

\tikzstyle{process} = [rectangle,  
draw=black,text width=1.5cm, text centered]

\tikzstyle{decision} = [diamond, 
draw=black,text width=1.5cm, text centered]
\tikzstyle{arrow} = [thick,->,>=stealth]
\usepackage{epsfig}
\usepackage{rotating}
\usepackage{amssymb}
\usepackage{xcolor}
\usepackage[T1]{fontenc}
\usepackage{lmodern}
\usepackage[utf8]{inputenc}
\usepackage{booktabs}
\usepackage{amsmath,amssymb}
\usepackage{multirow}
\usepackage[ruled,vlined]{algorithm2e}
\usepackage{caption}
\usepackage{subcaption}
\usepackage{lscape,float}
\usepackage{graphicx}
\usepackage{xcolor}
\usepackage{soul}
\usepackage{setspace}
\usepackage[toc]{appendix}

\usepackage{natbib,hyperref,doi}
\makeatletter
\renewcommand\@biblabel[1]{#1.}
\makeatother
\usepackage[noend]{algpseudocode}
\usepackage{enumitem}

\usepackage{amsmath}
\usepackage{amsthm}
\usepackage{enumitem}

\newtheorem{theorem}{Theorem}[section]

\newtheorem{lemma}[theorem]{Lemma}
\newtheorem{proposition}[theorem]{Proposition}
\makeatletter
\renewcommand\@biblabel[1]{#1.}
\makeatother

\usepackage{float,lscape}
 \providecommand{\keywords}[1] {
  \textit{Keywords:} #1}
\title{Optimal Design for Generalized Progressive Hybrid Censored Data via Constrained, Unconstrained, Compound, and Minimax Optimization}
\author{Rathin Das$^{1,2}$\footnote
    {Corresponding author, e-mail: rathindas65@gmail.com},~  Tanmay Sen$^2$, and Deepak Prajapati$^1$ \\
    $^1$ Decision Sciences Area, Indian Institute of Management Lucknow, Lucknow, India\\
    $^2$ SQC \& OR Unit, Indian Statistical Institute, Kolkata-700108, India}

\usepackage{siunitx}
\usepackage{amsthm}
\usepackage[english]{babel}
\newtheorem*{remark}{Remark}

\newtheorem*{rmk*}{Remark}
\date{}
\begin{document}

\maketitle

\begin{abstract}

This paper studies the optimal design of Type-I generalized progressive hybrid censoring schemes for life-testing experiments. The design problem involves simultaneously determining the inspection time, the guaranteed number of failures, and the progressive censoring scheme. First we develop a cost-constrained optimization framework for determining the optimal censoring scheme. Structural properties of the A-optimality criterion and the experimental cost with respect to the inspection time and the guaranteed number of failures are established. It reveals that they are conflicting behaviors which enables to develop an efficient search algorithm that substantially reduces the computational burden. Building on these theoretical results, a multi-objective optimization model is proposed to simultaneously minimize A-optimality criterion and the experimental cost. A Variable Neighborhood Search (VNS) algorithm is proposed to efficiently determine the optimal progressive removal vector by exploring the feasible design space while avoiding exhaustive enumeration. The resulting compromise designs simultaneously improve estimation precision and reduce experimental cost. In addition, the Shannon differential entropy of the observed lifetime distribution is derived and employed as a complementary information-theoretic measure for evaluating the selected censoring schemes. Numerical studies show that entropy-optimal designs generally differ from A-optimal designs, indicating that Shannon entropy characterizes uncertainty in the observed data rather than estimation precision. The proposed methodology provides an efficient computational framework for optimal life-test design and offers a foundation for future multi-objective optimization incorporating statistical efficiency, experimental cost, and information-theoretic uncertainty.

\end{abstract}

\section{Introduction}

Life-testing experiments play a central role in reliability engineering, quality control, and industrial product development, where the objective is to assess the lifetime characteristics of products and estimate reliability measures with limited experimental resources. However, due to time, cost and other resource limitations, censored life tests are conducted to collect lifetime information. Among the classical censoring mechanisms, Type-I censoring terminates the experiment at a predetermined time, whereas Type-II censoring terminates after observing a predetermined number of failures. Hybrid censoring schemes combine the advantages of both approaches by allowing the experiment to terminate according to either a specified time or a specified number of failures. The Type-I hybrid censoring scheme was introduced by \citet{Epstein_1954}. Comprehensive reviews of hybrid censoring schemes and their applications can be found in \citet{Bala_hybrid}.

Although hybrid censoring reduces experimental duration, it does not permit the removal of surviving units during the life test. To address this limitation, \citet{Cohen_1963} introduced the progressive Type-II censoring scheme (PCS-II), in which surviving units may be progressively withdrawn at each observed failure. Progressive censoring provides considerable operational flexibility by reducing maintenance costs and allowing testing facilities to be utilized more efficiently. However, when highly reliable products are tested, PCS-II experiments may still require an excessively long testing period. To overcome this difficulty, \citet{Kundu_2006} proposed the Type-I progressive hybrid censoring scheme (PHCS-I). This schemes successfully combine progressive censoring with hybrid termination rules, but each possesses inherent limitations. In particular, PHCS-I may terminate after observing only a small number of failures.

To improve the flexibility of PCHS-I, \citet{cho2015exact} proposed type-I generalized progressive hybrid censoring scheme (GPHCS-I). Under GPHCS-I, a minimum number of failures is guaranteed before termination while simultaneously restricting the experiment to a predetermined censoring time. Consequently, GPHCS-I provides a desirable compromise between obtaining sufficient failure information and controlling the experimental duration. When the guaranteed number of failures is zero, GPHCS-I reduces to the conventional PHCS-I scheme.

Since their introduction, GPHCS-I have attracted considerable attention in the statistical literature. Existing research has primarily focused on statistical inference, including parameter estimation, Bayesian inference, prediction, reliability analysis, and confidence interval construction under various lifetime distributions. Representative contributions include \citet{sen2018fisher}, \citet{singh2021inference}, \citet{elshahhat2022statistical}, \citet{singh2024partially}, \citet{hassan2024bayesian}, and \citet{xin2024statistical}. Despite these developments, comparatively little attention has been devoted to the design of optimal life-testing experiments under GPHCS-I. In practice, the selection of the sample size, censoring time, guaranteed number of failures, and progressive removal scheme has a direct influence on the statistical efficiency, information content, testing duration, and experimental cost. Consequently, determining an efficient life-testing plan naturally leads to a challenging optimization problem involving multiple conflicting objectives.

Optimal experimental design has traditionally been developed using information matrix-based criteria, such as D-optimality and A-optimality, together with various economic design criteria (\cite{shah2012theory}). Among these, the A-optimality criterion seeks to minimize the trace of the inverse Fisher information matrix and therefore minimizes the average variance of the parameter estimators. However, Fisher information quantifies only estimation precision and does not fully characterize the information contained in the observed experiment.

Information-theoretic measures provide a complementary perspective for evaluating experimental designs. Following the concept of expected Shannon information introduced by \citet{lindley1956measure}, entropy has become an important measure of information in Bayesian experimental design. Motivated by this idea, we derive the Shannon entropy under the GPHCS-I scheme and develop an entropy-based optimality criterion for life-testing experiments. Unlike Fisher information, Shannon entropy measures the uncertainty associated with the observed experiment and therefore provides an alternative assessment of the quality of an experimental design (see \citet{xu2015reference,sen2020inference}).

Besides statistical efficiency and information content, economic considerations are equally important in practical reliability experiments. Longer testing durations and larger numbers of observed failures generally improve estimation accuracy but inevitably increase the experimental cost. Therefore, practical experimental design should simultaneously consider statistical precision, information content, and economic efficiency. Such conflicting objectives naturally motivate the development of constrained and multi-objective optimization models.

For multi-objective optimization, \citet{bhattacharya2020implementation} first introduced the concept of compound optimal design in life-testing experiments. Subsequently, the methodology has been extended to various reliability models and censoring schemes by \citet{chakraborty2023cumulative, chakraborty2025application, dhameliya2024compound, dhameliya2025optimizing}. Although the compound optimality approach provides an effective framework for balancing conflicting objectives, its performance depends on the selection of preference weights and the normalization of the objective functions. Moreover, the monotonicity properties of the individual objectives are not directly incorporated into the optimization procedure, limiting the possibility of reducing the search space. In practice, the determination of appropriate weights often relies on graphical analyses of trade-off curves, which become increasingly computationally demanding as the number of objectives increases. In particular, extending the compound design framework beyond two objectives is challenging because graphical weight selection is no longer straightforward.

Motivated by these limitations, this paper proposes a minimax optimization framework for determining optimal GPHCS-I designs. The proposed approach exploits the monotonicity properties of the objective and cost functions to reduce the feasible search space. Therefore it improves computational efficiency. Furthermore, unlike the compound optimality approach, the minimax formulation can be readily extended to accommodate more than two conflicting objectives without requiring graphical weight selection or reformulation of the optimization problem. These features make the proposed framework computationally efficient and practically attractive for reliability design problems involving multiple competing objectives.

This paper develops a comprehensive optimization framework for GPHCS-I life-testing experiments. We first investigate unconstrained optimal designs based on A-optimality, entropy, and experimental cost. Constrained A-optimal and entropy-optimal designs are then developed under fixed budget limitations. Finally, we formulate a multi-objective optimization problem to simultaneously optimize statistical efficiency and experimental cost using both minimax and compound formulations. The minimax formulation is particularly attractive because it provides balanced compromise solutions, naturally extends to additional objectives such as entropy. 

Another important contribution of this paper is the development of theoretical properties that substantially simplify the optimization problem. We establish monotonicity results for the A-optimality criterion and the experimental cost with respect to the censoring time and the guaranteed number of failures. These analytical results enable a considerable reduction in the feasible search space and are incorporated into a variable neighborhood search algorithm, resulting in significant computational savings compared with exhaustive optimization.  The main contributions of this paper are summarized as follows.

\begin{enumerate}
\item We derive the Shannon entropy under the Type-I generalized progressive hybrid censoring scheme and propose an entropy-based criterion for optimal life-testing design.

\item We develop unconstrained optimal designs based on A-optimality, entropy, and experimental cost, together with constrained optimal designs under budget limitations.

\item We establish theoretical monotonicity properties of the A-optimality criterion and the experimental cost, leading to substantial reductions in the optimization search space.

\item We formulate multi-objective optimization models based on both minimax and compound criteria for simultaneously optimizing statistical efficiency and experimental cost.

\item We develop an efficient variable neighborhood search algorithm for determining optimal GPHCS-I designs under unconstrained, constrained, and multi-objective settings.

\item Extensive numerical studies under Weibull lifetime models demonstrate the effectiveness of the proposed optimization framework and illustrate the advantages of GPHCS-I over the conventional progressive hybrid censoring scheme.
\end{enumerate}

The remainder of the paper is organized as follows. Section~\ref{prob} introduces GPHCS-I and the associated statistical model. Section~\ref{design} derives the Fisher information matrix and the Shannon entropy and presents the corresponding optimality criteria. Section~\ref{proce} develops the proposed optimization formulations for the single-objective design problem together with the variable neighborhood search algorithm. Section~\ref{multi-obj} extends the methodology to the multi-objective setting by developing both minimax and compound optimization formulations. Section~\ref{numerical} presents comprehensive numerical studies for the unconstrained, constrained, and multi-objective optimal designs. Section~\ref{data} illustrates the proposed methodology through a real-life data analysis and provides the optimal design parameters under different experimental settings. Finally, Section~\ref{con} concludes the paper and discusses directions for future research.

\section{Problem  Statement}\label{prob}
The objective of this paper is to determine efficient life-testing designs under GPHCS-I. This section formulates the optimization problem by introducing the GPHCS-I model, the design variables, and the experimental cost function. These components are then integrated into a generic optimization framework, which serves as the foundation for the optimality criteria and optimization algorithms developed in the subsequent 
\subsection{Type-I Generalized Progressive Hybrid Censoring Scheme}
Consider a life-testing experiment in which $n$ identical and independent test units are placed on test under the GPHCS-I. Let $X_1,\ldots,X_n$ denote the lifetimes of the test units, each having common cumulative distribution function (CDF) $F(\cdot\mid\boldsymbol{\theta})$, probability density function (PDF) $f(\cdot\mid\boldsymbol{\theta})$, and hazard function $h(\cdot\mid\boldsymbol{\theta})$, where $\boldsymbol{\theta}$ denotes the unknown parameter vector. Let \(
X_{1:n}\le X_{2:n}\le\cdots\le X_{n:n}
\) be the corresponding order statistics. Let $m$ denote the maximum number of failures to be observed during the life-testing experiment, $T$ be the predetermined censoring time, and $l$ $(0\leq l<m)$ be the guaranteed minimum number of failures. Let \(
\mathbf{R}=(R_1,R_2,\ldots,R_m)\)
be the preassigned progressive removal vector satisfying \(\sum_{i=1}^{m}R_i=n-m.\)

Under the progressive censoring scheme, $n$ identical and independent test units are initially placed on test. At the time of the first observed failure, $R_1$ of the remaining $n-1$ surviving units are randomly withdrawn from the experiment. Subsequently, at the time of the second observed failure, $R_2$ of the remaining $n-R_1-2$ surviving units are removed. This progressive withdrawal mechanism continues according to the predetermined removal vector $\mathbf{R}$ until the experiment is terminated. Let \(
X_{1:m:n},X_{2:m:n},\ldots,X_{m:m:n}\)
denote the observed progressively censored failure times.

In GPHCS-I, the termination rule of the experiments as follows: If the censoring time occurs before the $l$th failure, that is, $T<X_{l:m:n}$, the experiment continues until the $l$th failure is observed and terminates at $X_{l:m:n}$. Otherwise, if the $l$th failure occurs before time $T$, the experiment proceeds until either the censoring time or the $m$th failure is reached. Consequently, if $X_{d:m:n}<T<X_{d+1:m:n}$ for some $l\le d\le m-1$, the experiment terminates at time $T$ after observing $d$ failures. On the other hand, if the $m$th failure occurs before the censoring time, that is, $X_{m:m:n}<T$, the experiment terminates immediately after the $m$th failure. The stopping time is therefore given by
\(\tau=\max\left\{X_{l:m:n},\min\left(X_{m:m:n},T\right)\right\}.\) It is evident that GPHCS-I reduces to the PHCS-I when $l=0$.
 The PDF and CDF of the $i$th progressively censored order statistic $X_{i:m:n}$ are
\begin{equation}
f_{i:m:n}(x\mid\boldsymbol{\theta})
=\sum_{k=1}^i {c_{k,i}}{N_{k-1}}
(1-F(x\mid\boldsymbol{\theta}))^{N_{k-1}-1}f(x\mid\boldsymbol{\theta}),
\end{equation}
and
\begin{equation}
F_{i:m:n}(x\mid\boldsymbol{\theta})
=
1-\sum_{k=1}^i {c_{k,i}}
(1-F(x\mid\boldsymbol{\theta}))^{N_{k-1}},
\end{equation}
respectively, where
\[
N_{j-1}=n-\sum_{k=1}^{j-1} R_k-(j-1),
\quad
c_{k,i}=\prod_{\substack{j=1\\j\neq k}}^i
\frac{N_{j-1}}{N_{j-1}-N_{k-1}}.
\]
Let $D$ denote the total number of failures observed before termination. Under the GPHCS-I scheme, the pair $(D,\tau)$ is given by
\[
(D,\tau)=
\begin{cases}
(l,X_{l:m:n}), & T<X_{l:m:n},\\
(D_1,T), & X_{d_1:m:n}<T<X_{d_1+1:m:n}, ~l\leq d_1\leq m-1\\
(m,X_{m:m:n}), & X_{m:m:n}<T.
\end{cases}
\]
\subsection{Decision Variables and Cost Model}
For fixed values of the sample size $n$ and maximum number of observed failure $m$, the objective is to determine the optimal GPHCS-I design by selecting the progressive removal vector, the censoring time, and the guaranteed minimum number of failures. Accordingly, the design decision vector is defined as \( \boldsymbol{\zeta}=(\boldsymbol{R},T,l)\), where $\boldsymbol{R}=(R_1,\ldots,R_m)$. Let $\varphi(\boldsymbol{\zeta})$ be generic the optimality criteria, which is minimized or maximized to determine the optimal value of $\boldsymbol{\zeta}$.

Now, life testing must also account for an economic constraint. Therefore, the experimenter must design the life test under budget constraints. Let $TC(\boldsymbol{\boldsymbol{\zeta}})$ denote the total cost of running a life test under a GPHCS-I. Note that the total cost depends on the following costs as follows:
\begin{enumerate}[label=\Roman*.]
    \item Operation Cost: This involves the cost of running a life test, including the cost of utilities and the salary of the operators. Let $\tau$ denote the duration of the life test. Furthermore, let $C_\tau$ be the operating cost per unit of time. Then the resulting total operating cost is $C_\tau E[\tau]$, where $E[\tau]$ is the expected duration of the PIC-I scheme and is given by
  \[
E[\tau]=E[X_{l:m:n}]+E[X_{m:m:n}\wedge T]-E[X_{l:m:n}\wedge T].\]
    where \[E[X_{i:m:n}\wedge T]=\int_0^{T}(1-F_{i:m:n}(x\mid\boldsymbol{\theta}))\,dx.\]

\item Failure Cost: This refers to the expenses incurred from the failure of a test unit during a life test experiment, encompassing the costs associated with identifying the failure mode, performing a rework, scrapping the defective unit, and any additional procedures necessary to understand the failure. Let $C_D$ be the cost associated with a failed test unit. So, the resulting total cost is given by $C_DE[D]$, where $E[D]$ is the expected number of failures and is given by 
\begin{align*}
    E[D]=l[1-F_{l:m:n}(T\mid\boldsymbol{\theta})]+\sum_{d=l}^{m-1}d\Pr(X_{d:m:n}<T<X_{d+1:m:n})+mF_{m:m:n}(T\mid\boldsymbol{\theta}).
\end{align*}
For $l<d<m$, 
\begin{align*}
    &\Pr(X_{d:m:n}<T<X_{d+1:m:n})\\=& C_{d-1}\bigl(1 - F(T\mid\boldsymbol{\theta})\bigr)^{N_d}
\sum_{i=0}^{d-1}
C_{i,d-1}(R_1+1,\ldots,R_{d-1}+1)
\frac{
1 - \bigl(1 - F(T\mid\boldsymbol{\theta})\bigr)^{R'_{i0}}
}{
R'_{i0}
},
\end{align*}
where $C_{l-1} = n (n - R_1 - 1)\cdots (n - \sum_{i=1}^{l-1} R_i - (l-1))$, \(
C_{i,r}(a_1,\ldots,a_r)=\frac{(-1)^i}{\left[\prod_{j=1}^{i}
\sum\limits_{k=r-i+j}^{r}
a_k\right]\left[
\prod\limits_{j=1}^{r-i}
\sum_{k=j}^{r-i}
a_k\right]
}
\), \(R'_{i0}=\sum_{j=d-i}^{d} (R_j + 1)\) and \(R''_{0}=n - d - \sum_{j=1}^{d} R_j=N_d.\)
\end{enumerate}
Thus, the total cost corresponding to a life test is given by
\begin{align*}
    TC(\boldsymbol{\boldsymbol{\zeta}})=C_0+C_\tau E[\tau]+C_DE[D].
\end{align*}
Let $C_B$ be the available budget for conducting the life test. Then the general constrained optimization problem can be formulated as
\begin{align}
\max_{\boldsymbol{\zeta}}\text{ or }\min_{\boldsymbol{\zeta}}\quad
& \varphi(\boldsymbol{\zeta}) \nonumber\\
\text{subject to}\quad
& TC(\boldsymbol{\zeta})\le C_B,\nonumber\\
& \sum_{i=1}^{m}R_i=n-m,\nonumber\\
& 0\le l\le m-1,\nonumber\\
& R_i,l\in\mathbb{N}\cup\{0\},\qquad i=1,\ldots,m.
\label{eq:general_problem}
\end{align}

%
\section{Design Criteria}\label{design}
This section introduces the optimality criteria considered in this paper for determining efficient GPHCS-I life-testing designs. Two complementary criteria are employed. The first is the A-optimality criterion, which is based on the Fisher information matrix and aims to improve the precision of parameter estimation. The second is an information-theoretic criterion based on Shannon entropy, which quantifies the uncertainty associated with the observed experiment. The expression of the Fisher information matrix is given in Section \ref{fisher}.  The derivation of the entropy function is given in  Section \ref{ent}. 
\subsection{Fisher information under GPHCS-I}\label{fisher}
The Fisher information matrix under the GPHCS-I scheme was derived by \citet{sen2018fisher}. For completeness, the required expression is summarized below.
\[
I_{X_{l:m:n}\vee(X_{m:m:n}\wedge T)}(\boldsymbol\theta)
=
I_{1,\ldots,l:m:n}(\boldsymbol\theta)
+
I_{X_{m:m:n}\wedge T}(\boldsymbol\theta)
-
I_{X_{l:m:n}\wedge T}(\boldsymbol\theta),
\]
with
\[
I_{1,\ldots,r:m:n}(\boldsymbol\theta)
=
\int_0^\infty
\left[
\nabla_{\boldsymbol{\theta}}
\ln h(x\mid\boldsymbol\theta)
\right]\left[
\nabla_{\boldsymbol{\theta}}
\ln h(x\mid\boldsymbol\theta)
\right]^T
\sum_{i=1}^r f_{i:m:n}(x\mid\boldsymbol\theta)\,dx,
\]
\[
I_{X_{m:m:n}\wedge T}(\boldsymbol\theta)
=
\int_0^{T}
\left[
\nabla_{\boldsymbol{\theta}}
\ln h(x\mid\boldsymbol\theta)
\right]\left[
\nabla_{\boldsymbol{\theta}}
\ln h(x\mid\boldsymbol\theta)
\right]^T
\sum_{i=1}^m f_{i:m:n}(x\mid\boldsymbol{\theta})\,dx,
\]
\[
I_{X_{l:m:n}\wedge T}(\boldsymbol\theta)
=
\int_0^{T}
\left[
\nabla_{\boldsymbol{\theta}}
\ln h(x\mid\boldsymbol\theta)
\right]\left[
\nabla_{\boldsymbol{\theta}}
\ln h(x\mid\boldsymbol\theta)
\right]^T
\sum_{i=1}^l f_{i:m:n}(x\mid\boldsymbol{\theta})\,dx,
\]
where $\nabla_{\boldsymbol{\theta}}
A(\boldsymbol{\theta})=\left(\frac{\partial}{\partial\theta_1}A(\boldsymbol{\theta}),\ldots,\frac{\partial}{\partial\theta_k}A(\boldsymbol{\theta})\right)$.
Based on the Fisher information matrix derived in Section~\ref{fisher}, the A-optimality criterion is defined as
\[
\phi(\boldsymbol{\zeta})
=
\operatorname{tr}
\left\{
I^{-1}(\boldsymbol{\theta};\boldsymbol{\zeta})
\right\},
\]
where $\operatorname{tr}(\cdot)$ denotes the trace operator. Since the trace of the inverse Fisher information matrix represents the total variance of the parameter estimators, the A-optimal design is obtained by solving
\[
\min_{\boldsymbol{\zeta}}
\;
\phi(\boldsymbol{\zeta}).
\]

\subsection{Shannon Entropy under GPHCS-I}\label{ent}
Although the Fisher information matrix provides an effective measure of the precision of parameter estimation, it does not directly quantify the information content of the observed censored sample. An alternative measure is the Shannon entropy, which characterizes the uncertainty associated with the observed data and has been widely adopted as an information-theoretic criterion in optimal experimental design. For continuous random variables, the Shannon entropy may be either positive or negative depending on the underlying distribution and its parameters. Following \citet{lindley1956measure}, the negative entropy can be interpreted as a measure of the information contained in the observed sample. Consequently, minimizing the Shannon entropy is equivalent to maximizing the information extracted from the life-testing experiment. Therefore, the Shannon entropy provides a natural criterion for designing efficient life-testing experiments under censoring schemes (see also \citet{xu2015reference,sen2020inference}).

 The following theorem establishes the Shannon entropy of the GPHCS-I by exploiting the decomposition property of progressively censored order statistics.







\begin{theorem}
The entropy under GPHCS-I data is given by
\(
\mathcal{H}_{X_{l:m:n} \vee (X_{m:m:n} \wedge T)} = \mathcal{H}_{1,\ldots,l:m:n} + \mathcal{H}_{X_{m:m:n} \wedge T} - \mathcal{H}_{X_{l:m:n} \wedge T},
\)
wwhere $\mathcal{H}_{1,\ldots,l:m:n}$ denotes the entropy under PCS-II, and $\mathcal{H}_{X_{m:m:n}\wedge T}$ and $\mathcal{H}_{X_{l:m:n}\wedge T}$ denote the entropies under PHCS-I.

\end{theorem}
\begin{proof}
Let us define
\[
Y_i = 
\begin{cases}
X_{i:m:n} & \text{for } i = 1, 2, \ldots, l, \\
\left(X_{i:m:n} \wedge T,\ \mathbb{I}(X_{i:m:n} \leq T)\right) & \text{for } i = l+1, \ldots, m,
\end{cases}
\]
where \( \mathbb{I}(\cdot) \) denotes the indicator function.

Then the joint density of GPHCS-I data $(Y_1, \ldots, Y_l,$ $ Y_{l+1}, \ldots, Y_m)$ is decomposed by using Markov property of order statistics as 
\small\begin{eqnarray}
	\label{ent1a}
	f_{1,\ldots,l,l+1,\ldots,m:m:n}(y_1,\ldots,y_l,y_{l+1},\ldots,y_{m}\mid\boldsymbol{\theta}) 
	=f_{1,\ldots,l:m:n}(y_1,\ldots,y_{l}\mid\boldsymbol{\theta})\times f_{l+1,\ldots,m|l:m:n}(y_{l+1},\ldots,y_{m}|y_l\mid\boldsymbol{\theta}). 
\end{eqnarray}
\normalsize
We denote the entropy of the joint densities  $f_{1,\ldots,l:m:n}(y_1,\ldots,y_{l}\mid\boldsymbol{\theta})$ and $ f_{l+1,\ldots,m|l:m:n}(y_{l+1},\ldots,y_{m}|y_l\mid\boldsymbol{\theta})$ are  \( \mathcal{H}_{1,\ldots,l:m:n} \) and \( \mathcal{H}_{l+1,\dots,m \mid l:m:n} \) respectively. Hence,
\[
\mathcal{H}_{X_{l:m:n}\vee(X_{m:m:n}\wedge T)}
=
\mathcal{H}_{1,\ldots,l:m:n}
+
\mathcal{H}_{l+1,\ldots,m\,|\,l:m:n},
\]
where $\mathcal{H}_{1,\ldots,l:m:n}$ represents the entropy under PCS-II and is given by (\citet{balakrishnan2007testing})
\[
\mathcal{H}_{1,\ldots,l:m:n} = -\ln(C_{l-1}) + l- \left( \int_0^\infty \sum_{i = 1}^{l}f_{i:m:n}(x) \ln h(x) \, dx \right).
\]

To derive the entropy $\mathcal{H}_{l+1,\dots,m \mid l:m:n}$, we need to decompose the conditional joint density function $f_{l+1,\ldots,m|l:m:n}(y_{l+1},\ldots,y_{m}|y_l\mid\boldsymbol{\theta})$ as \
\begin{eqnarray*}
    f_{l+1,\ldots,m|l:m:n}(y_{l+1},\ldots,y_{m}|y_l\mid\boldsymbol{\theta}) = \prod_{j=l+1}^{m}f_{j|j-1:m:n}(y_j|y_{j-1}\mid\boldsymbol{\theta}),
\end{eqnarray*}
where $f_{j \mid j-1:m:n}(y_j \mid y_{j-1}; \theta)$ represents the conditional density of $Y_{j:m:n}$ given $Y_{j-1:m:n} = y_{j-1}$. The entropy of $f_{l+1,\dots,m \mid l:m:n}(y_{l+1}, \dots, y_m; \theta)$ can be written as 
\allowdisplaybreaks\begin{eqnarray*}
    \mathcal{H}_{l+1,\dots,m \mid l:m:n} &=& \sum_{j=l+1}^{m} H_{j:m:n \mid j-1:m:n}\\
    &=& \sum_{j=l+1}^{m} \mathcal{H}_{Z_{1:N_{j-1}} \wedge T}\\
    &=& \sum_{j=1}^{m} \mathcal{H}_{Z_{1:N_{j-1}} \wedge T}-\sum_{j=1}^{l} \mathcal{H}_{Z_{1:N_{j-1}} \wedge T}\\
     &=&\mathcal{H}_{X_{m:m:n} \wedge T} - \mathcal{H}_{X_{l:m:n} \wedge T},
\end{eqnarray*}
where $\mathcal{H}_{X_{m:m:n} \wedge T}$ and $\mathcal{H}_{X_{l:m:n} \wedge T}$ represent the entropies under PCS-II and is given by (\citet{almohaimeed2017entropy})
\allowdisplaybreaks\begin{eqnarray}
 \mathcal{H}_{X_{i:m:n} \wedge T}
&=& \sum_{j=1}^{m}\left(1 - \log N_{j-1} \right) 
\left(1 -  \binom{N_{j-1}}{j-1} \, {(1-F(T))}^{N_{j-1}+1} \right) 
- \int_0^{T} \log h(t) \, f_{j:m:n}(t) \, dt. \label{eqn_phcs}
\end{eqnarray}
Then from (\ref{ent1a}), we derive the entropy under GPHCS-I as
\begin{eqnarray*}
\mathcal{H}_{X_{l:m:n} \vee (X_{m:m:n} \wedge T)}&=&{\mathcal{H}}_{1,\cdots,l:m:n}+ \mathcal{H}_{l+1,\dots,m \mid l:m:n}\\
	&=&{\mathcal{H}}_{1,\cdots,l:m:n}+\mathcal{H}_{X_{m:m:n} \wedge T} - \mathcal{H}_{X_{l:m:n} \wedge T}
\end{eqnarray*}
\end{proof}
 Therefore, the entropy-optimal design is obtained by solving
\[
\min_{\boldsymbol{\zeta}}
\;
\mathcal{\phi}(\boldsymbol{\zeta})=\mathcal{H}_{X_{l:m:n} \vee (X_{m:m:n} \wedge T)}.
\]
\section{Procedure for Cost-constraint Optimization}\label{proce}
This section develops the optimization methodology for determining the A-optimal and entropy-optimal GPHCS-I designs. We first establish several theoretical properties of the total cost of the experiment and the A-optimality criterion. These results reveal the monotonic behavior of the objective and cost functions with respect to the design variables and are subsequently exploited to substantially reduce the computational burden of the optimization procedure. Based on these properties, efficient algorithms are proposed to determine the optimal GPHCS-I designs under budget constraints.

To investigate the monotonicity properties of the total expected experimental cost, it is first necessary to study the behavior of the expected test duration and the expected number of observed failures. The following propositions establish the monotonicity properties of these quantities, which are subsequently used to derive the monotonicity of the total expected cost and to develop efficient optimization algorithms.
\begin{proposition}
    For fixed $(\boldsymbol{R},l)$ the expected duration under GPHCS-I is increasing in $T$
\end{proposition}\label{r1}
    \begin{proof}
    We know that $E[\tau]=E[X_{l:m:n}]+E[X_{m:m:n}\wedge T]-E[X_{l:m:n}\wedge T].$ Now we want to prove that \(E[\tau]\) is increasing in \(T\). Using the identity
\[E[X\wedge t]=\int_0^{t}\Pr(X>x)\,dx,
\]
we obtain
\begin{align*}
  E[\tau]&=E[X_{l:m:n}]+\int_0^{T}\Pr(X_{m:m:n}>x)\,dx-\int_0^{T}\Pr(X_{l:m:n}>x)\,dx\\
  &=E[X_{l:m:n}]+\int_0^{T}\left[\Pr(X_{m:m:n}>x)-\Pr(X_{l:m:n}>x)\right]dx.
\end{align*}
Differentiating with respect to \(T\),
\[\frac{d}{dT}E[\tau]=\Pr(X_{m:m:n}>T)-\Pr(X_{l:m:n}>T).\]
Since $X_{l:m:n}$ and $X_{m:m:n}$ are progressively censored order statistics, the $l$-th observed failure time can never exceed the $m$-th observed failure time whenever $l\leq m$. Therefore \(\Pr(X_{m:m:n}>T)\geq\Pr(X_{l:m:n}>T).\) Hence,
\[
\frac{d}{dT}E[\tau]\geq 0.
\] Therefore,
\(
E[\tau]
\)
is increasing in \(T\).

\end{proof}
\begin{proposition}\label{r2}
    For fixed $(\boldsymbol{R},T)$, the expected duration is increasing in $l$.
\end{proposition}
\begin{proof}
The expected duration can be written as $E[\tau]=E[X_{m:m:n}\wedge T]+g(l),$ where \(
g(l)=E[X_{l:m:n}]-E[X_{l:m:n}\wedge T]. \) It is seen that the first term is independent of \(l\). Hence, it is sufficient to show that \(g(l)\) is increasing in \(l\). Now $g(l)$ can be written as
\begin{align*}
    g(l)=\int_{T}^\infty P(X_{l:m:n}>x)~ dx
\end{align*}
Now we want to prove that $g(l_1)\leq g(l_2)$, for \(l_1<l_2\). 
We know that for $l_1<l_2$, $P(X_{l_1:m:n}>x)\leq P(X_{l_2:m:n}>x)$, for all $x$. Then 
\begin{align*}
   &\int_T^\infty P(X_{l_1:m:n}>x)~dx\leq \int_{T}^\infty P(X_{l_2:m:n}>x)~dx
\end{align*}
Assume that both integrals are finite. This completes the proof.
\end{proof}
\begin{proposition}\label{r3}
  \begin{enumerate}
      \item For fixed $(\boldsymbol{R},l)$, the expected number of failures is increasing in $T$.
      \item For fixed $(\boldsymbol{R},T)$, the expected number of failures is increasing in $l$. 
  \end{enumerate}  
\end{proposition}
\begin{proof}
    We know that \(E[D]=l\bigl[1-F_{l:m:n}(T\mid\boldsymbol{\theta})\bigr]+\sum_{d=l+1}^{m-1}d\,\Pr(X_{d:m:n}<T<X_{d+1:m:n})+mF_{m:m:n}(T\mid\boldsymbol{\theta}).\) We want to prove that \(E[D]\) is increasing in \(T\). It is observed that
\[
\Pr(X_{d:m:n}<T<X_{d+1:m:n})
=
F_{d:m:n}(T\mid\boldsymbol{\theta})
-
F_{d+1:m:n}(T\mid\boldsymbol{\theta}).
\]
Hence,
\[
\begin{aligned}
E[D]
&=l\bigl[1-F_{l:m:n}(T\mid\boldsymbol{\theta})\bigr]+\sum_{d=l}^{m-1}d\left[F_{d:m:n}(T\mid\boldsymbol{\theta})-F_{d+1:m:n}(T\mid\boldsymbol{\theta})\right]+mF_{m:m:n}(T\mid\boldsymbol{\theta})\\
&=l+\sum_{d=l}^{m}F_{d:m:n}(T\mid\boldsymbol{\theta}).
\end{aligned}
\]
Since CDF is increasing in $T$. Therefore, it is clearly shows that $E[D]$ is increasing in $T$. Also, $E[D]$ is clearly increasing in $l$.

\end{proof}
Using propositions~\ref{r1}--\ref{r3}, the following monotonicity properties of the total expected experimental cost are obtained.
\begin{theorem}\label{th3}
   The total expected cost \( TC(\boldsymbol{\zeta}) \) under the GPHCS-I scheme satisfies the following monotonicity properties.
    \begin{enumerate}[label=\Roman*]
        \item For fixed $(\boldsymbol{R},l)$, $\phi(\boldsymbol{\boldsymbol{\zeta}})$ is increasing in $T$.
        \item For fixed $(\boldsymbol{R},T)$, $\phi(\boldsymbol{\boldsymbol{\zeta}})$ is increasing in $l$.     
    \end{enumerate}
\end{theorem}
\begin{theorem}\label{th1}
    The design criterion $\phi(\boldsymbol{\boldsymbol{\zeta}})=\text{trace}[I^{-1}(\boldsymbol{\theta}\mid\boldsymbol{\zeta})]$ satisfies the following monotonicity properties:
    \begin{enumerate}[label=\Roman*]
        \item For fixed $(\boldsymbol{R},l)$, $\phi(\boldsymbol{\boldsymbol{\zeta}})$ is decreasing in $T$.
        \item For fixed $(\boldsymbol{R},T)$, $\phi(\boldsymbol{\boldsymbol{\zeta}})$ is decreasing in $l$.     
    \end{enumerate}
\end{theorem}
\begin{proof}
\textbf{Proof of part (I): }  Let $\boldsymbol{\boldsymbol{\zeta}}_1=(\boldsymbol{R},T,l)$ and $\boldsymbol{\boldsymbol{\zeta}}_2=(\boldsymbol{R},T,l+1)$.  Then, it is enough to show that $\phi(\boldsymbol{\boldsymbol{\zeta}}_2)\leq \phi(\boldsymbol{\boldsymbol{\zeta}}_1)$.
    \begin{align*}
    &I(\boldsymbol{\theta}\ | \ \boldsymbol{\boldsymbol{\zeta}}_2)-I(\boldsymbol{\theta}\ | \ \boldsymbol{\boldsymbol{\zeta}}_1)=\int_0^{T}
\left[
\nabla_{\boldsymbol{\theta}}
\ln h(x\mid\boldsymbol\theta)
\right]\left[
\nabla_{\boldsymbol{\theta}}
\ln h(x\mid\boldsymbol\theta)
\right]^T
 f_{(l+1):m:n}(x\mid\boldsymbol{\theta})\,dx.
\end{align*}
Now, for any  $\boldsymbol{a}(\boldsymbol{\theta})\neq 0$, we have
\begin{align*}
   & \boldsymbol{a}(\boldsymbol{\theta})^T[I(\boldsymbol{\theta}\ | \ \boldsymbol{\boldsymbol{\zeta}}_2)-I(\boldsymbol{\theta}\ | \ \boldsymbol{\boldsymbol{\zeta}}_1)]\boldsymbol{a}(\boldsymbol{\theta}) =\int_0^{T}
\left[a(\boldsymbol{\theta})^T
\nabla_{\boldsymbol{\theta}}
\ln h(x\mid\boldsymbol\theta)
\right]^2
 f_{(l+1):m:n}(x\mid\boldsymbol{\theta})\,dx.
\end{align*}
This shows that $I(\boldsymbol{\theta}\ | \ \boldsymbol{\boldsymbol{\zeta}}_2)-I(\boldsymbol{\theta}\ | \ \boldsymbol{\boldsymbol{\zeta}}_1)$ is a non-negative definite matrix. This implies that $I^{-1}(\boldsymbol{\theta}\ | \ \boldsymbol{\boldsymbol{\zeta}}_1)-I^{-1}(\boldsymbol{\theta}\ | \ \boldsymbol{\boldsymbol{\zeta}}_2)$ is also non-negative definite matrix. Therefore, we have
\begin{align*}
    &\text{trace}\left[I^{-1}(\boldsymbol{\theta}\ | \ \boldsymbol{\boldsymbol{\zeta}}_1)-I^{-1}(\boldsymbol{\theta}\ | \ \boldsymbol{\boldsymbol{\zeta}}_2)\right]\geq 0.
\end{align*}
This completes the proof of part (I).

\noindent\textbf{Proof of part (II):} Let  $\boldsymbol{\boldsymbol{\zeta}}_1=(\boldsymbol{R},l,T_1)$ and $\boldsymbol{\boldsymbol{\zeta}}_2=(\boldsymbol{R},l,T_2)$, where $T_1<T_2$.  Then it is enough to prove that $\phi(\boldsymbol{\boldsymbol{\zeta}}_1)\leq \phi(\boldsymbol{\boldsymbol{\zeta}}_2)$.
\begin{align*}
   & I(\boldsymbol{\theta}\ | \ \boldsymbol{\boldsymbol{\zeta}}_1)-I(\boldsymbol{\theta}\ |\ \boldsymbol{\boldsymbol{\zeta}_2})=\int_{T_1}^{T_2}
\left[
\nabla_{\boldsymbol{\theta}}
\ln h(x\mid\boldsymbol\theta)
\right]\left[
\nabla_{\boldsymbol{\theta}}
\ln h(x\mid\boldsymbol\theta)
\right]^T
\sum_{i=l+1}^m f_{i:m:n}(x\mid\boldsymbol{\theta})\,dx.
\end{align*}
Now, for any  $\boldsymbol{a}(\boldsymbol{\theta})\neq 0$, we have
\begin{align*}
   & \boldsymbol{a}(\boldsymbol{\theta})^T[I(\boldsymbol{\theta}\ | \ \boldsymbol{\boldsymbol{\zeta}}_2)-I(\boldsymbol{\theta}\ | \ \boldsymbol{\boldsymbol{\zeta}}_1)]\boldsymbol{a}(\boldsymbol{\theta})=\int_{T_1}^{T_2}
\left[a(\boldsymbol{\theta})^T
\nabla_{\boldsymbol{\theta}}
\ln h(x\mid\boldsymbol\theta)
\right]^2
\sum_{i=l+1}^m f_{i:m:n}(x\mid\boldsymbol{\theta})\,dx.
\end{align*}
This shows that $I(\boldsymbol{\theta}\ | \ \boldsymbol{\boldsymbol{\zeta}}_2)-I(\boldsymbol{\theta}\ | \ \boldsymbol{\boldsymbol{\zeta}}_1)$ is a non-negative definite matrix. This implies that $I^{-1}(\boldsymbol{\theta}\ | \ \boldsymbol{\boldsymbol{\zeta}}_1)-I^{-1}(\boldsymbol{\theta}\ | \ \boldsymbol{\boldsymbol{\zeta}}_2)$ is also non-negative definite matrix. Therefore, we have
\begin{align*}
    &\text{trace}\left[I^{-1}(\boldsymbol{\theta}\ | \ \boldsymbol{\boldsymbol{\zeta}}_1)-I^{-1}(\boldsymbol{\theta}\ | \ \boldsymbol{\boldsymbol{\zeta}}_2)\right]\geq 0
\end{align*}
This completes the proof of part (II).
\end{proof}
Theorems~\ref{th3} and \ref{th1} provide the theoretical basis for developing an efficient optimization algorithm. For a fixed pair $(\boldsymbol{R},l)$, Theorem~\ref{th1} shows that the A-optimality criterion is decreasing in the censoring time $T$, whereas Theorem~\ref{th3} establishes that the total expected cost is increasing in $T$. Consequently, whenever the budget constraint is active, the optimal censoring time is attained at the boundary of the feasible region and can be obtained by solving \(TC(\boldsymbol{\zeta})=C_B.\)

The monotonicity of the total expected cost also provides an effective feasibility test. Since $T=0$ and $T=\infty$ cannot be used directly in numerical computation, we replace them by sufficiently small and sufficiently large finite values, denoted by $T_{\min}$ and $T_{\max}$, respectively. For a fixed progressive censoring vector $\boldsymbol{R}$, if
\(
TC(\boldsymbol{R},T_{\min},0)>C_B,
\)
then the budget constraint is violated even at the smallest admissible censoring time. Hence, no feasible pair $(T,l)$ exists for that particular $\boldsymbol{R}$, and the corresponding neighborhood is discarded from the VNS search.

Furthermore, for a fixed $\boldsymbol{R}$, if
\(
TC(\boldsymbol{R},T_{\min},l)>C_B\)
for some value of $l$, then, by Theorem~\ref{th3}, all subsequent values $l+1,\ldots,m-1$ are also infeasible. Therefore, the search over $l$ can be terminated immediately. On the other hand, if
\(
TC(\boldsymbol{R},T_{\max},l)\le C_B,
\)
then the budget constraint remains satisfied throughout the search interval. Since the A-optimality criterion decreases with $T$, the optimal censoring time is simply given by
\(
T^*=T_{\max}.
\)

Therefore, for each progressive censoring vector, only a one-dimensional search over $l$ is required. The resulting optimization problem is thus reduced to determining the optimal progressive censoring vector, which is efficiently accomplished using the Variable Neighborhood Search (VNS) algorithm described in Algorithm~\ref{algo1}.
\begin{algorithm}[hbt!]
\small\caption{VNS Algorithm for Constrained A-optimal Design}
\label{algo1}

Initialize the feasible progressive censoring vector \(
\boldsymbol{R}=(0,0,\ldots,0,n-m).
\)

\For{each feasible progressive censoring vector $\boldsymbol{R}$}{
    Set $l=0$\;
    \While{$l\le m-1$}{
        \If{$TC(\boldsymbol{R},T_{\min},l)>C_B$}{
            \textbf{break}
        }
        \eIf{$TC(\boldsymbol{R},T_{\max},l)\le C_B$}{
            Set \(T^*(\boldsymbol{R},l)=T_{\max}.\)
        }
        {
        Determine
\(T^*(\boldsymbol{R},l)\) by solving \(TC(\boldsymbol{R},T,l)=C_B.\)
        }
        Compute \(\phi(\boldsymbol{R},T^*,l).\)

        Store the objective value\;

        Set $l=l+1$\;

    }

    Determine
    \(l^*(\boldsymbol{R})=\arg\min_l\phi(\boldsymbol{R},T^*,l).\)
    
    Define
    \(
    \phi_1(\boldsymbol{R})=\phi(\boldsymbol{R},T^*,l^*(\boldsymbol{R})).\)

}

Apply the VNS algorithm to minimize
\(\phi_1(\boldsymbol{R})\)
and obtain the optimal
\(
\boldsymbol{R}^*.
\)

Obtain
\(
l^*=l^*(\boldsymbol{R}^*),~
T^*=T^*(\boldsymbol{R}^*,l^*).
\)

\Return{$(\boldsymbol{R}^*,T^*,l^*)$.}

\end{algorithm}

Unlike the A-optimality criterion, the Shannon entropy does not satisfy any monotonicity property with respect to the censoring time $T$ or the guaranteed minimum number of failures $l$. Consequently, the theoretical properties established for the A-optimality criterion cannot be exploited to eliminate candidate solutions directly. Nevertheless, the monotonicity of the total expected cost established in Theorem~\ref{th3} remains valid and can still be utilized to reduce the feasible search region. Specifically, for a fixed progressive censoring vector $\boldsymbol{R}$, if
\(
TC(\boldsymbol{R},T_{\min},0)>C_B,\)
then no feasible pair $(T,l)$ exists for that particular $\boldsymbol{R}$, and the corresponding neighborhood is discarded. Similarly, if
\(
TC(\boldsymbol{R},T_{\min},l)>C_B
\)
for some value of $l$, then all subsequent values of $l$ are also infeasible because the total expected cost is increasing in $l$. Hence, the search over $l$ can be terminated immediately. Furthermore, if
\(
TC(\boldsymbol{R},T_{\max},l)\le C_B,
\)
the entire interval $[T_{\min},T_{\max}]$ is feasible, and only this finite interval needs to be explored. Also, For a fixed progressive censoring vector $\boldsymbol{R}$ and a fixed value of $l$, let $T_B$ denote the unique solution of
\(
TC(\boldsymbol{R},T,l)=C_B.
\)
Since the total expected cost is increasing in $T$, every feasible censoring time satisfies
\(
T\in[T_{\min},\,T_B].
\)
Therefore, instead of searching over the entire interval $[T_{\min},T_{\max}]$, it is sufficient to search over the reduced interval $[T_{\min},T_B]$.

Therefore, although the entropy criterion itself does not possess monotonicity, the monotonicity of the total expected cost substantially reduces the feasible search space. The resulting optimization procedure is summarized in Algorithm~\ref{algo2}.

\begin{algorithm}[H]
\small\caption{VNS Algorithm for Constrained Entropy-optimal Design}
\label{algo2}

Initialize the feasible progressive censoring vector
\(
\boldsymbol{R}=(0,0,\ldots,0,n-m).
\)

\For{each feasible progressive censoring vector $\boldsymbol{R}$}{

    Set $l=0$\;

    \While{$l\le m-1$}{

        \If{$TC(\boldsymbol{R},T_{\min},l)>C_B$}{
\textbf{break}
        }
Determine $T_B$ by solving
        \(TC(\boldsymbol{R},T,l)=C_B.
        \)

        \For{$T\in[T_{\min},T_B]$}{
Compute \(\phi(\boldsymbol{R},T,l).\)
        }

        Determine
        \(T^*(\boldsymbol{R},l)=\arg\min_{T\in[T_{\min},T_B]}
        \phi(\boldsymbol{R},T,l).
        \)
        
        Store
        \(
        \phi(\boldsymbol{R},T^*,l).\)
        
        Set $l=l+1$\;
    }
    Determine
    \(l^*(\boldsymbol{R})=\arg\min_l\phi_E(\boldsymbol{R},T^*,.
    \)
   
    Define
    \(\phi_1(\boldsymbol{R})= \phi_E(\boldsymbol{R},T^*(\boldsymbol{R},l^*),l^*(\boldsymbol{R})).
    \)

}
Apply the VNS algorithm to minimize
\(
\phi_1(\boldsymbol{R})
\)
and obtain the optimal
\(
\boldsymbol{R}^*.
\)

Obtain
\(
l^*=l^*(\boldsymbol{R}^*),~
T^*=T^*(\boldsymbol{R}^*,l^*).
\)

\Return{$(\boldsymbol{R}^*,T^*,l^*)$.}

\end{algorithm}

\section{Multi-objective Design}\label{multi-obj}
The results obtained in the section \ref{proce} show that improving the statistical efficiency of a life-testing experiment generally requires a longer test duration and a larger number of observed failures, thereby increasing the total experimental cost. Consequently, the experimenter must balance two conflicting objectives: minimizing the A-optimality criterion, which represents the estimation uncertainty, and minimizing the total cost of the experiment. This naturally leads to a multi-objective optimization problem.

In this section, we consider two approaches for obtaining compromise designs. The first is the well-established compound optimality criterion, which combines the normalized objectives through a weighted sum. The second is a minimax criterion, which minimizes the largest normalized objective. The latter avoids the need to specify preference weights and naturally exploits the monotonicity properties established in the previous section, leading to a more computationally efficient optimization procedure. Moreover, the minimax formulation can be readily extended to problems involving more than two competing objectives.  

\subsection{Compound Design}
Let $\phi_1(\boldsymbol{\zeta})$ denote the A-optimality criterion and let $\phi_2(\boldsymbol{\zeta})$ denote the total experimental cost under GPHCS-I/ Therefore, $\phi_1(\boldsymbol{\zeta})=TC(\boldsymbol{\zeta})$ and $\phi_2(\boldsymbol{\zeta})=\phi(\boldsymbol{\zeta})$. Let $\boldsymbol{\zeta}_1^*=(R_1^*,T_1^*,l_1^*)$ and $\boldsymbol{\zeta}_2^*=(\boldsymbol{R}_2^*,T_2^*,l_2^*)$ be optimal designs of $\phi_1(\boldsymbol{\zeta})$ and $\phi_2(\boldsymbol{\zeta})$, respectively. The compound criterion is defined as
\[
\psi(\boldsymbol{\zeta} \mid k)=\lambda \bar\phi_1(\boldsymbol{\zeta})+(1-\lambda)\bar\phi_2(\boldsymbol{\zeta}),
\]
where
\[
\bar\phi_i(\boldsymbol{\zeta})=\frac{\phi_i(\boldsymbol{\zeta})}{\phi_i(\boldsymbol{\zeta}_i^*)}, \qquad i=1,2,
\]
represent the relative scaling associated with the respective criteria. Since both $\phi_1(\boldsymbol{\zeta})$ and $\phi_2(\boldsymbol{\zeta})$ are positive, then construction, the relative scaling satisfy
\(\bar\phi_i(\boldsymbol{\zeta}) \geq 1.
\). Thus, the reciprocal
\[
\frac{1}{\bar{\phi}_i(\boldsymbol{\zeta})}
=
\frac{\phi_i(\boldsymbol{\zeta}_i^*)}
{\phi_i(\boldsymbol{\zeta})},
\]
lies in the interval $[0,1]$ and represents the efficiency of the $i$th objective relative to its individual optimum. An efficiency value closer to one indicates that the corresponding design performs closer to the ideal solution for that objective.

A detailed discussion regarding the theoretical properties and applications of the compound optimal design framework may be found in \citet{bhattacharya2020implementation}. We now have two competing criteria and must choose an experiment that simultaneously minimizes the trace of the inverse Fisher information matrix and the total cost. Now, for every value of $0\leq k\leq 1$, we minimize the compound criteria $\psi(\boldsymbol{\zeta}\mid \lambda)$. Let for a fixed value of \(\lambda\),
\(\boldsymbol{\zeta}_\lambda=\arg\min\limits_{\boldsymbol{\zeta}}
\psi(\boldsymbol{\zeta}\mid \lambda)
\)
denote the design that minimizes the compound criterion $\boldsymbol{\psi}(\boldsymbol{\zeta}\mid \lambda)$.
The procedure of finding $\boldsymbol{\zeta}_\lambda$ is discussed using Algorithm \ref{algo}.

\begin{algorithm}[hbt!]
\caption{VNS Algorithm for Determining the Compound Design}
\label{algo}
Choose a value of the preference parameter $\lambda\in[0,1]$\;
Initialize the feasible progressive censoring vector
\(
\boldsymbol{R}=(0,0,\ldots,0,n-m).
\)

\For{each feasible progressive censoring vector $\boldsymbol{R}$}{

    \For{$l=0,1,\ldots,m-1$}{

        \For{$T\in[T_{\min},T_{\max}]$}{

            Compute
            \(
            \psi(\boldsymbol{R},T,l) = \lambda\bar{\phi}_1(\boldsymbol{R},T,l)+ (1-\lambda)\bar{\phi}_2(\boldsymbol{R},T,l).
            \)

        }

        Determine
        \(
        T^*(\boldsymbol{R},l)=\arg\min_T\psi(\boldsymbol{R},T,l).\)

        Store
        \(\psi(\boldsymbol{R},T^*,l).\)

    }

    Determine
    \(l^*(\boldsymbol{R})= \arg\min_l \psi(\boldsymbol{R},T^*,l).\)

    Define
    \(
    \psi_1(\boldsymbol{R})= \psi(\boldsymbol{R},T^*(\boldsymbol{R},l^*), l^*(\boldsymbol{R})).
    \)
}

Apply the VNS algorithm to minimize
\(
\psi_1(\boldsymbol{R})\)
and obtain the optimal \(
\boldsymbol{R}^*.\)

Return
\((\boldsymbol{R}^*,T^*,l^*).
\)

\end{algorithm}

\begin{theorem}
The designing parameter $\boldsymbol{\zeta}_\lambda$ satisfies the following property
\begin{enumerate}
    \item the corresponding optimal value
\(
\phi_1(\boldsymbol{\zeta}_\lambda)
\)
is increasing as \(\lambda\) increases.
\item the corresponding optimal value
\(
\phi_2(\boldsymbol{\zeta}_\lambda)
\)
is decreasing as \(\lambda\) increases.
\end{enumerate}
\end{theorem}
\begin{proof}
    The proof can be followed from \cite{dhameliya2024compound}.
\end{proof}
From the Theorem, we can say that if we plot \(
1/\phi_1(\boldsymbol{\zeta}_\lambda)
\) and \(
1/\phi_2(\boldsymbol{\zeta}_\lambda)
\) for different values of $\lambda$ between 0 and 1, there can exist a point $\lambda^*$ where $
\phi_1(\boldsymbol{\zeta}_{\lambda^*})=
\phi_1(\boldsymbol{\zeta}_{\lambda^*})$. Now using the value of $\lambda^*$ and Algorithm \ref{algo}, we get the optimal compound design $\boldsymbol{\zeta}^*$ by minimizing $\psi(\boldsymbol{\zeta}\mid \lambda^*)$.

Although the compound criterion provides an effective compromise between statistical efficiency and experimental cost, its implementation requires repeated optimization for different values of $\lambda$. The final design is then selected by analysing the trade-off curve to determine an appropriate balancing weight. Furthermore, extending this procedure beyond two objectives becomes increasingly difficult because graphical determination of the preference weights is no longer straightforward.
\subsection{Minimax criterion}
\label{sec:minimax}

To overcome these limitations, we propose a minimax optimization framework. Rather than specifying preference weights, the minimax criterion minimizes the largest normalized deviation from the individual optimum. Consequently, it produces a balanced design without requiring any prior preference information.  The minimax compromise design is defined by
\begin{equation}
(\mathbf{R}^*,T^*,l^*)
=
\arg\min_{\mathbf{R},T,l}
\left\{
\max
\left(
\phi_1(\mathbf{R},T,l),
\phi_2(\mathbf{R},T,l)
\right)
\right\},
\label{eq:minimax}
\end{equation}
To incorporate the monotonicity properties established in Section~\ref{proce} into the multi-objective optimization framework, we first establish the following lemma and theorem. These results considerably reduce the feasible search space and improve the computational efficiency of the proposed optimization Algorithm \ref{algo:minimax} to determine the optimal designing parameter $\boldsymbol{\zeta}^*$.
\begin{lemma}
\label{lem:mono}

For fixed inspection level $l$ and progressive replacement vector
$\mathbf{R}$, define

\[
g(T)
=
\phi_1(\mathbf{R},T,l)-\phi_2(\mathbf{R},T,l).
\]
Using Theorems \ref{th3} and \ref{th1}, we can say that $g(T)$ is strictly increasing over $T$.

\end{lemma}







\begin{theorem}
\label{thm:minimax}

For fixed inspection level $l$ and progressive replacement vector
$\mathbf{R}$, consider

\[
\Phi(\boldsymbol{R},T,l)
=
\max
\left\{
\phi_1(\mathbf{R},T,l),
\phi_2(\mathbf{R},T,l)
\right\}.
\]
Then,
\begin{enumerate}

\item
If there exists a feasible point $T^*$ satisfying
\(\phi_1(\mathbf{R},T^*,l)=\phi_2(\mathbf{R},T^*,l),
\)
then
\(
T^*=\arg\min\limits_T\phi(\mathbf{R},T^*,l).
\)

\item
If \(\phi_1(\mathbf{R},T,l)<\phi_2(\mathbf{R},T,l),\) for all $T$ then \(T^*\rightarrow\infty.\)
\item
If
\(\phi_1(\mathbf{R},T,l)>\phi_2(\mathbf{R},T,l),\) for all $T$ then
\(T^*\rightarrow0.\)

\end{enumerate}

\end{theorem}


\begin{proof}
By Lemma~\ref{lem:mono}, $g(T)$ is strictly increasing.
Suppose there exists $T^*$ satisfying \(
g(T^*)=0.
\) Then
\[
\phi_1(l,T^*,\mathbf{R})=\phi_2(l,T^*,\mathbf{R})=c.
\]
For every $T<T^*$, \(\phi_1(\boldsymbol{R},T,l)<c,~\phi_2(\mathbf{R},T,l)>c.\) Hence,

\[
\Phi(\mathbf{R},T,l)
=
\phi_2(l,T,\mathbf{R})
>
c.
\]
Similarly, for every
$T>T^*$,
\(\phi_1(\mathbf{R},T,l)>c,~\phi_2(\mathbf{R},T,l)<c,
\)
which gives
\[
\Phi(\mathbf{R},T,l)
=
\phi_1(l,T,\mathbf{R})
>
c.
\]
Therefore,
\(
\Phi(\mathbf{R},T,l)
\ge
\Phi(\mathbf{R},T^*,l)
\) for every feasible $T$ and fixed $l$ and $R$, proving part (1).

If \(g(T)<0\) throughout the feasible interval, then \(
\phi(\mathbf{R},T,l)=
\phi_2(\mathbf{R},T,l),\) which is decreasing in $T$. Hence, \(
T^*\rightarrow\infty.,\) proving part (2).

Similarly, if \(g(T)>0\) throughout the feasible interval, then
\(\phi(\mathbf{R},T,l)=\phi_1(\mathbf{R},T,l),\) which is increasing in $T$. Hence, \(
T^*\rightarrow0,\) , proving part (3).
\end{proof}







\begin{algorithm}[hbt!]
\caption{Determination of the Minimax Compromise Design}\label{algo:minimax}
Initialize the feasible progressive replacement vector
\(
\boldsymbol{R}^{(0)}.
\)
\For{a fixed progressive replacement vector \(\boldsymbol{R}\)}{
Set \(l=0\)\;
\While{$l\le m-1$}{
Compute \(g(T)=\phi_1(\boldsymbol{R},T,l)-\phi_2(\boldsymbol{R},T,l).\)

\eIf{$g(T_{\min})\,g(T_{\max})\le0$}{
Determine \(T^*(l,\boldsymbol{R})\) by solving \(
\phi_1(\boldsymbol{R},T,l)=\phi_2(\boldsymbol{R},T,l)\) using the bisection method.
}{
\eIf{$g(T_{\min})>0$}{
Set
\(
T^*(l,\boldsymbol{R})=T_{\min}.\)

}{

Set

\(
T^*(l,\boldsymbol{R})=T_{\max}.
\)
}
}
Compute
\(\Phi(l,\boldsymbol{R})=\max\left\{\phi_1(\boldsymbol{R},l,T^*),\phi_2(\boldsymbol{R},l,T^*)\right\}.\)

Increase \(l\)\;
}
Determine
\(
l^*(\boldsymbol{R})
=\arg\min_l\Phi(l,\boldsymbol{R}).\)
Define
\(
\Phi_1(\boldsymbol{R})=\Phi\left(l^*(\boldsymbol{R}),\boldsymbol{R}\right).\)
}
Apply the VNS algorithm to minimize
\(
\Phi_1(\boldsymbol{R})
\) and obtain the optimal
\(
\boldsymbol{R}^*.
\)

Obtain \(l^*=l^*(\boldsymbol{R}^*)\)
and \(T^*=T^*(l^*,\boldsymbol{R}^*).\)

Return the optimal compromise solution \(
(\boldsymbol{R}^*,\,T^*,\,l^*).
\)

\end{algorithm}


\begin{remark}
The proposed algorithm exploits the monotonicity of the objective
functions with respect to the inspection interval $T$.
Consequently, the continuous optimization problem is reduced to either a one-dimensional root-finding problem or a boundary evaluation. $l$ is then searched over their feasible domains and $\boldsymbol{R}$ is searched using VNS algorithm. This substantially reduces the computational burden compared with an
exhaustive grid search over the continuous decision variable.
\end{remark}

\section{Numerical Analysis}\label{numerical}
In this section, we illustrate the proposed optimal design methodology for the Type-I generalized progressive hybrid censoring scheme (GPHCS-I) under several design criteria. We first determine the unconstrained optimal designs for the A-optimality, cost-minimization, and entropy-minimization criteria in Section \ref{uncons}. We then investigate the constrained A-optimal and entropy-optimal designs under a fixed budget in Section \ref{cons}. Finally, the multi-objective design, which simultaneously considers A-optimality and cost minimization, is presented in Section \ref{multi}. Throughout the numerical study, it is assumed that the lifetime of each test item follows a Weibull distribution with cumulative distribution function (CDF)
\[F(x\mid\boldsymbol{\theta})=1-\exp\left[-(\eta x)^{\alpha}\right],\]
which is denoted by \(\text{Wei}(\alpha,\eta)\), where \(\alpha>0\) and \(\eta>0\) are the shape and scale parameters, respectively.

\subsection{Unconstrained optimal design}\label{uncons}
For the A-optimal design, Result \ref{r3} shows that the optimal censoring time satisfies \(T^* \to \infty\) and the optimal guaranteed number of failures is \(l^*=m-1\).  In contrast, when the objective is to minimize the total experimental cost, Results \ref{r1} and \ref{r2} imply that the optimal design satisfies \(T^* \to 0\) and \(l^*=0\). Consequently, the theoretical minimum cost is the fixed setup cost \(C_0\), since no effective testing time or failure observation is required.  However, these theoretical optima are not computationally practical because \(T=0\) and \(T=\infty\) cannot be implemented numerically. Therefore, throughout the numerical study, the search space for the censoring time is restricted to \(T \in [T_{\min},T_{\max}]\), where \(T_{\min}\) and \(T_{\max}\) denote the 0.01th and 0.99th quantiles of the lifetime distribution, respectively. Under this restriction, the cost-minimization problem always attains its optimum at the lower bound \(T^*=T_{\min}\), whereas the A-optimal design always attains its optimum at the upper bound \(T^*=T_{\max}\). Based on these theoretical results, the optimal censoring scheme \(\boldsymbol{R}^*\) is obtained using the VNS algorithm for \(\text{Wei}(2,1)\), \(\text{Wei}(1,1)\), and \(\text{Wei}(0.5,1)\).  The corresponding A-optimal for different combinations of \(n\) and \(m\) are reported in Table \ref{uncontraint}. Also it is noted that in every cases, for cost minimization, the optimal censoring scheme \(\boldsymbol{R}^*=(n-m,0^{(m-1)}))\), which is not given in the table. $0^{(m-1)})$ means that 0 repated in $(m-1)$ times. Also, the optimal entropy minimization is given in Table \ref{uncontraint}. To determine the optimal designing parameter for the unconstrained entropy minimization, we can use Algorithm \ref{algo} by taking $\psi(\boldsymbol{R},T,l)=\mathcal{H}_{X_{l:m:n} \vee (X_{m:m:n} \wedge T)}$. The optimal objective value is denoted by $\phi_1$ in Table \ref{uncontraint} for unconstrained settings. For the Weibull distributions \(\text{Wei}(2,1)\), \(\text{Wei}(1,1)\), and \(\text{Wei}(0.5,1)\), the corresponding values of \(T_{\min}\) are \(0.0001\), \(0.0101\), and \(0.1002\), respectively, while the corresponding values of \(T_{\max}\) are \(21.2076\), \(4.6052\), and \(2.1460\), respectively.
\begin{table}[hbt!]
\centering
\caption{Unconstrained optimal design}
\begin{tabular}{ccc|ccc|ccc}
\hline
 &&  & \multicolumn{3}{c|}{A-optimal} & \multicolumn{3}{c}{Entropy}\\
\hline
$\alpha$&$n$ &$m$&$\boldsymbol{R}^*$&$(T^*,l^*)$&$\phi_1$&$\boldsymbol{R}^*$&$(T^*,l^*)$&$\phi_1$
\\
\hline
\multirow{6}{*}{0.5}&10&5&$(3,1^{(2)},0^{(2)})$&$(21.2076,4)$&0.8302&$(0^{(4)},5)$&$(0.556,3)$&-19.7741\\

&15&5&$(10,0^{(4)})$&$(21.2076,4)$&0.8270&$(0^{(4)},10)$&$(8.113,3)$&-29.0091\\
&20&5&$(15,0^{(4)})$&$(21.2076,4)$&0.8263&$(0^{(4)},15)$&$(12.849,3)$&-35.5443\\
&15&7&$(6,1,0^{(4)},1)$&$(21.2076,6)$&0.5903&$(0^{(6)},8)$&$(8.571,6)$&-40.1336\\
&15&10&$((1,0)^{(4)},0,1)$&$(21.2076,9)$&0.4156&$(0^{(9)},5)$&$(0.271,7)$&-53.3251\\
&15&12&$(0^{(11)},3)$&$(21.2076,11)$&0.3495&$(0^{(11)},3)$&$(0.264,7)$&-58.8179\\
\hline
\multirow{6}{*}{1}&10&5&$(0,5,0^{(3)})$&$(4.6052,4)$&0.2957&$(0^{(4)},5)$&$(4.601,4)$&-5.3169\\
&15&5&$(10,0^{(4)})$&$(4.6052,4)$&0.2823&$(0^{(4)},10)$&$(3.521,2)$&-7.7948\\
&20&5&$(15,0^{(4)})$&$(4.6052,4)$&0.2741&$(0^{(4)},15)$&$(2.850,0)$&-9.4363\\
&15&7&$(0,8,0^{(5)})$&$(4.6052,6)$&0.2078&$(0^{(6)},8)$&$(3.938,0)$&-10.2947\\
&15&10&$(0^{(3)},5,0^{(6)})$&$(4.6052,9)$&0.1544&$(0^{(9)},5)$&$(2.967,9)$&-13.1120\\
&15&12&$(0^{(4)},3,0^{(7)})$&$(4.6052,11)$&0.1344&$(0^{(11)},3)$&$(2.967,9)$&-14.1075\\
\hline
\multirow{6}{*}{2}&10&5&$(0,5,0^{(3)})$&$(2.1460,4)$&0.4269&$(2,0^{(3)},3)$&$(2.146,4)$&-2.8000\\
&15&5&$(0,10,0^{(3)})$&$(2.1460,4)$&0.3674&$(7,0^{(3)},3)$&$(2.146,4)$&-2.7351\\
&20&5&$(0,15,0^{(3)})$&$(2.1460,4)$&0.3321&$(12,0^{(3)},3)$&$(2.146,4)$&-2.8580\\
&15&7&$(0,8,0^{(5)})$&$(2.1460,6)$
&0.2910&$(6,0^{(5)},2)$&$(2.146,6)$&-5.4268\\
&15&10&$(0^{(2)},5,0^{(7)})$&$(2.1460,9)$&0.2339&$(3,0^{(8)},2)$&$(2.146,9)$&-10.8369\\
&15&12&$(0^{(2)},3,0^{(9)})$&$(2.1460,11)$&0.2102&$(2,0^{(10)},1)$&$(2.146,11)$&-15.2842\\
\hline
\end{tabular}
\label{uncontraint}
\end{table}

 The optimal censoring time and the guaranteed number of failures vary with the underlying Weibull distribution and the design parameters, indicating that the entropy criterion is not monotonic with respect to the design variables. Furthermore, increasing the sample size or the effective sample size generally produces larger absolute entropy values, indicating that the experiment contains richer information as more observations become available.
\subsection{Constrained optimal design}\label{cons}
The cost components are fixed at \(C_D=10\), \(C_\tau=50\), and \(C_0=50\), while the budget cost is taken as \(C_B=100\), \(125\), and \(150\). The resulting constrained A-optimal designs for the progressive hybrid censoring scheme are presented in Tables \ref{alpha=0.5}-\ref{m}.  Tables \ref{alpha=0.5}-\ref{alpha=2} present the optimal designs for Weibull distributions with shape parameters \(\alpha=0.5\), \(1\), and \(2\), respectively, when \(m=5\). The results are reported for different sample sizes \((n=10,15,\) and \(20)\) and budget costs \((C_B=100,125,\) and \(150)\). Table 7 presents the optimal designs for \(\alpha=0.5\) and \(n=15\), considering different effective sample sizes \((m=7,10,\) and \(12)\) and budget costs. In Tables~\ref{alpha=0.5}--\ref{m}, we report the optimal designs and the corresponding optimal objective values obtained under the GPHCS-I and PHCS-I schemes for both the A-optimality and entropy criteria. Throughout these tables, $\phi$ denotes the optimal objective value under the GPHCS-I scheme, whereas $\phi_{2}$ denotes the corresponding optimal objective value under the PHCS-I scheme. Therefore, we have \(
\phi_2 \leq \phi \leq \phi_1.
\)
 
To assess the performance of the proposed constrained A-optimal design and entropy design, we consider two performance measures. 
The first performance measure is the relative risk saving (RRS), which quantifies the percentage reduction in estimation risk achieved by the GPHCS-I relative to the PHCS-I, and is defined as
\[
\mathrm{RRS}
=
\frac{\phi_1-\phi}{|\phi_1|}\times100\%.
\]
A larger value of RRS indicates a greater improvement in estimation efficiency achieved by the GPHCS-I over the PHCS-I.
The second performance measure is the relative efficiency loss (REL), which quantifies the loss in estimation efficiency due to the budget constraint relative to the unconstrained A-optimal design. It is defined as
\[
\mathrm{REL}
=
\frac{\phi-\phi_2}{|\phi_2|}\times100\%.
\]
A smaller value of REL indicates that the constrained design is closer to the unconstrained A-optimal design and, therefore, incurs only a small loss in estimation efficiency due to the budget constraint. The values of RRS and REL are reported in Tables~\ref{alpha=0.5}-\ref{m} to compare the constrained GPHCS-I with the constrained PHCS-I and to assess the efficiency loss of the constrained design relative to the unconstrained A-optimal design.
  

\begin{table}[hbt!]
    \centering
 \small    \caption{Optimal design for different values of $C_B$ and $n$ when $\alpha=0.5$ and $m=5$}
    \begin{tabular}{|c|cc|ccc|ccc|cc|}
    \hline
 &  &&\multicolumn{3}{c}{GPHCS-I}& \multicolumn{3}{|c|}{PHCS-I-I}&\multirow{2}{*}{RRS}&\multirow{2}{*}{REL}\\
   \cline{2-9}
   &$n$&$C_B$&$\boldsymbol{R}^*$&$(T^*,l^*)$&$\phi$&$\boldsymbol{R}^*$&$T^*$&$\phi_1$&&\\
      \hline
    \multirow{9}{*}{\begin{turn}{90}
    A-optimal
\end{turn}}&  &100&(0, 0, 0, 0, 5)&(0.215, 3)& 1.6859&(0, 0, 0, 0, 5)&0.2806&1.7249&2.26 & 103.07\\
   &  10 &125&(0, 0, 0, 0, 5)&(2.954, 4)&0.9782&(0, 0, 0, 0, 5)&3.0272& 0.9783&0.01 & 17.83\\    
    & &150&(2, 0, 0, 0, 3)&(21.207, 0)&0.8842&(2, 0, 0, 0, 3)&21.207&0.8842&0.00 & 6.50\\
      \cline{2-11}
    & &100&(3, 1, 0, 2, 4)&(0.0274, 4)&1.7199&(5, 0, 0, 0, 5)&0.2615&1.8065&4.79 & 107.97\\
   &15&125&(5, 0, 0, 1, 4)&(2.433, 4)& 1.0429&(5, 0, 0, 1, 4)&2.474& 1.0430&0.01 & 26.11\\
    & &150&(7, 0, 0, 1, 2)&(4.598, 4)&0.9221& (7, 0, 0, 1, 2)&4.6329&0.9222&0.01 & 11.50\\
      \cline{2-11}
     &&100&(9, 0, 0, 0, 6)&(0.0371, 4)&1.7548&(10, 0, 0, 0, 5)&0.253&1.8486&5.07 & 112.37\\
  & 20  &125&(11,  0,  0,  0,  4)&(1.2515, 4)&1.0693&(11, 0, 0, 0, 4)&1.3820&1.0703&0.09 & 29.41\\
   &  &150&(12,  0,  0,  1,  2)&(21.207, 0)&0.9412&(12,  0,  0,  1,  2)&21.207&0.9412&0.00 & 13.91\\
         \hline
   \multirow{6}{*}{\begin{turn}{90}
    Entropy
\end{turn}}  & & 100&(0, 0, 0, 0, 5)&(0.2805, 0)&-19.1080&(0, 0, 0, 0, 5)&0.2805&-19.1080&0.00 & 3.37\\
   & 10 &125&(0, 0, 0, 0, 5)&(0.5594, 3)&-19.7741&(0, 0, 0, 0, 5)&0.5949&-19.7638&0.05 & 0.00\\
   &&150&(0, 0, 0, 0, 5)&(0.5594, 3)&-19.7741&(0, 0, 0, 0, 5)&0.5949&-19.7638&0.05 & 0.00\\
      \cline{2-11}
    &   & 100&(0, 0, 0, 6, 4)&(0.1179, 4)&-25.1733&(0, 0, 0, 0, 10)&0.1834&-24.7853&1.57 & 13.22\\
    &15 &125&(0, 0, 0, 0, 10)&(4.8051,0)&-29.0091&(0, 0, 0, 0, 10)&4.8051&-29.0091&0.00 & 0.00\\
   &&150&(0, 0, 0, 0, 10)&(4.8051,0)&-29.0091&(0, 0, 0, 0, 10)&4.8051&-29.0091&0.00 & 0.00\\
     \cline{2-11}
    & &100&(0,  0,  0, 11,  4)&(0.1275, 4)&-29.6055&(0, 0, 0, 11, 4)&0.1537&-28.4296&4.14 & 16.71\\
   &20  &125&(0, 0, 0, 0, 15)&(12.8489,3)&-35.5443&(0, 0, 0, 0, 15)&12.8636&-35.5443&0.00 & 0.00\\
    & &150&(0, 0, 0, 0, 15)&(12.8489,3)&-35.5443&(0, 0, 0, 0 ,15)&12.8636&-35.5443&0.00 & 0.00\\
     \hline
 
    \end{tabular}
    \label{alpha=0.5}
\end{table}
\begin{table}[hbt!]
    \centering
   \small  \caption{Optimal design for different values of $C_B$ and $n$ when $\alpha=1$ and $m=5$}  
    \begin{tabular}{|c|cc|ccc|ccc|cc|}
    \hline  & &&\multicolumn{3}{c}{GPHCS-I}& \multicolumn{3}{|c|}{PHCS-I-I}&\multirow{2}{*}{RRS}&\multirow{2}{*}{REL}\\
   \cline{2-9}
   &$n$&$C_B$&$\boldsymbol{R}^*$&$(T^*,l^*)$&$\phi$&$\boldsymbol{R}^*$&$T^*$&$\phi_1$&&\\
      \hline
\multirow{9}{*}{\begin{turn}{90}
    A-optimal
\end{turn}}&     &100&(1, 1, 0, 0, 3)&(0.1424, 3)&0.7626&(1, 0, 0, 0, 4)&0.4085&0.9313&18.11 & 157.89\\
   & 10  &125&(2, 1, 0, 0, 2)&(0.265, 4)&0.4540&(0,  0, 0, 0, 5)&0.7914&0.4618&1.69 & 53.53\\
&    &150&(2, 0, 1, 0, 2)&(4.605, 4)&0.3357&(2, 0, 1, 0, 2)&4.6052&0.3357&0.00 & 13.53\\
      \cline{2-11}
&&100&(7, 0, 0, 0, 3)&(0.1877, 3)&0.7607&(6, 0, 0, 0, 4)&0.3888&0.9128&16.67 & 169.47\\
 &    15&125&(3, 0, 0, 0, 7)&(3.780, 0)&0.4286&(3, 0, 0, 0, 7)&4.6052&0.4286&0.00 & 51.82\\
  &   &150&(8, 0, 0, 0, 2)&(2.126, 4)& 0.3204&(8, 0, 0, 0, 2)&2.1588&0.3205&0.03 & 13.50\\
     \cline{2-11}
   &  &100&(12,  1,  0,  0,  2)&(0.131, 3)&0.7551&(11, 0, 0, 0, 4)&0.3794&0.8985&15.96 & 175.48\\
  &20   &125&(9, 0, 0, 0, 6)&(0.891, 4)&0.4218&(9, 0, 0, 0, 6)&0.9364&0.4224&0.14 & 53.88\\
   &  &150&(13,  0,  0,  0,  2)&(4.605, 4)&0.3101&(13,  0, 0, 0, 2)&4.6052&0.3101&0.00 & 13.13\\
     \hline
\multirow{9}{*}{\begin{turn}{90}
   Entropy
\end{turn}}    &  & 100&(0, 0, 3, 0, 2)&(0.2589, 3)&-3.6408&(0, 0, 2, 0, 3)&0.4006&-2.9386&23.89 & 31.53\\
  &  10 &125&(0, 0, 0, 0, 5)&(0.6769, 4)&-5.2130&(0, 0, 0, 0, 5)&0.7913&-5.1827&0.58 & 1.95\\
  & &150&(0, 0, 0, 0, 5)&(4.6012, 4)&-5.3169&(0, 0, 0, 0, 5)&0.4604&-5.3169&0.00 & 0.00\\
    \cline{2-11}
   &    & 100&(0, 0, 8, 0, 2)&(0.3049, 3)&-5.0606&(0, 0, 8, 0, 2)&0.3453&-4.0152&26.04 & 35.08\\
  &  15 &125&(0, 0, 0, 0, 10)&(3.5205, 2)&-7.7948&(0, 0, 0, 0, 10)&3.4803&-7.7949&0.00 & 0.00\\
  & &150&(0, 0, 0, 0, 10)&(3.5205, 2)&-7.7948&(0, 0, 0, 0, 10)&3.4803&-7.7949&0.00 & 0.00\\
  \cline{2-11}
   &  &100&(0,  0,  13, 0, 2)&(0.2985, 3)&-5.9680&(0, 0, 13, 0, 2)&0.3134&-5.1699&15.44 & 36.76\\
   &20  &125&(0, 0, 0, 0, 15)&(2.8500, 0)&-9.4363&(0, 0, 0, 0, 15)&2.8500&-9.4363&0.00 & 0.00\\
    & &150&(0, 0, 0, 0, 15)&(2.8500, 0)&-9.4363&(0, 0, 0, 0, 15)&2.8500&-9.4363&0.00 & 0.00\\
     \hline
    \end{tabular}
    \label{alpha=1}
\end{table}
\begin{table}[hbt!]
    \centering
  \small  \caption{Optimal design for different values of $C_B$ and $n$ when $\alpha=2$ and $m=5$}
    \begin{tabular}{|c|cc|ccc|ccc|cc|}
    \hline
& &&\multicolumn{3}{c}{GPHCS-I}& \multicolumn{3}{|c|}{PHCS-I-I}&\multirow{2}{*}{RRS}&\multirow{2}{*}{REL}\\
   \cline{2-9}
   &$n$&$C_B$&$\boldsymbol{R}^*$&$(T^*,l^*)$&$\phi$&$\boldsymbol{R}^*$&$T^*$&$\phi_1$&&\\
      \hline    
\multirow{9}{*}{\begin{turn}{90}
   A-optimal
\end{turn}}     & &100&(5, 0, 0, 0, 0)&(0.4334, 2)&1.3409&(4, 0, 0, 0, 1)&0.5891&1.7909&25.13 & 214.10\\
    &  10&125&(5, 0, 0, 0, 0)&(0.748, 3)&0.8035&(4, 0, 0, 0, 1)&0.8569&0.9291&13.52 & 88.22\\
   & &150&(3, 0, 0, 0, 2)&(2.146, 4)&0.5428&(1, 0, 1, 2, 1)&2.1460&0.5619&3.40 & 27.15\\
       \cline{2-11}
     &&100&(10,  0,  0,  0,  0)&(0.4624, 2)&1.210&(9, 0, 0, 0, 1)&0.5733&1.5491&21.89 & 229.34\\
    & 15&125&(10,  0,  0,  0,  0)&(0.7571, 3)&0.6948&(9, 0, 0, 0, 1)&0.8476&0.7889&11.93 & 89.11\\
     &&150&(8, 0, 0, 0, 2)&(2.146, 4)&0.4691&(7, 0, 0, 2, 1)&1.7898&0.4798&2.23 & 27.68\\
    \cline{2-11}
    & &100&(15,  0,  0,  0,  0)&(0.4705, 2)&1.1176&(14, 0, 0, 0, 1)&0.5657&1.4029&20.34 & 236.52\\
   &  20&125&(15,  0,  0,  0,  0)&(0.7603, 3)&0.6278&(14, 0, 0, 0, 1)&0.8432&0.7050&10.95 & 89.04\\
  &   &150&(13,  0,  0,  0,  2)&(2.146, 4)&0.4241&(13, 0, 0, 0, 2)&0.2160&0.4241&0.00 & 27.70\\
     \hline
\multirow{9}{*}{\begin{turn}{90}
   Entropy
\end{turn}}     &  & 100&(0, 5, 0, 0, 0)&(0.1003, 2)&0.7586&(4,  0, 0, 0, 1)&0.5891&1.1102&31.67 & 127.09\\\
   & 10 &125&(0, 0, 0, 3, 2)&(0.1003,4)&-1.8054&(2, 0, 0, 0, 3)&0.7952&-1.2309&46.67 & 35.52\\
  & &150&(2, 0, 0, 0, 3)&(2.1459, 4)&-2.800&(2, 0, 0, 0, 3)&2.1459&-2.8000&0.00 & 0.00\\
      \cline{2-11}
   &    & 100&(3, 6, 0, 0, 1)&(0.5029, 2)&0.5244&(0, 9, 0, 0, 1)&0.5206&0.9086&42.28 & 119.17\\
   & 15 &125&(6, 0, 0, 2, 2)&(0.1003, 4)&-1.9316&(7, 0, 0, 0, 3)&0.7836&-1.2973&48.89 & 29.38\\
  & &150&(7, 0, 0, 0, 3)&(2.1459, 4)&-2.7951&(7, 0, 0, 0, 3)&2.1459&-2.7951&0.00 & 2.19\\
     \cline{2-11}
   &  &100&(6, 9, 0, 0, 0)&(0.1003, 2)&0.3408&(0, 14, 0, 0, 1)&0.5006&0.6464&47.28 & 111.92\\
  & 20  &125&(11, 0,0, 2, 2)&(0.1003,4)&-2.0405&(12, 0, 0, 0, 3)&0.7785&-1.3897&46.83 & 28.60\\
   &  &150&(12, 0, 0, 0, 3)&(2.1459, 4)&-2.8580&(12, 0, 0, 0, 3)&2.1459&-2.8580&0.00 & 0.00\\
  \hline
    \end{tabular}
    \label{alpha=2}
\end{table}
\begin{table}[hbt!]
    \centering
 \small   \caption{Optimal design for different values of $C_B$ and $m$ when $\alpha=2$ and $n=15$}
    \begin{tabular}{|c|cc|ccc|ccc|cc|}
    \hline
 & &&\multicolumn{3}{c}{GPHCS-I}& \multicolumn{3}{|c|}{PHCS-I-I}&\multirow{2}{*}{RRS}&\multirow{2}{*}{REL}\\
   \cline{2-9}
   &$m$&$C_B$&$\boldsymbol{R}^*$&$(T^*,l^*)$&$\phi$&$\boldsymbol{R}^*$&$T^*$&$\phi_2$&&\\
      \hline    
\multirow{9}{*}{\begin{turn}{90}
   A-optimal
\end{turn}}&  &100&$(8, 0^{(7)})$&(0.4819, 2)&1.3361&$(8, 0^{(7)})$&0.5506&1.5573&14.20 & 359.14\\
    &  7&125&$(6, 0^{(5)}, 2)$&(0.4765, 4)&0.7065&$(8, 0^{(7)})$&0.7897&0.8103&12.81 & 142.78\\
   & &150&$(8, 0^{(7)})$&(0.7248, 5)&0.4683&$(8, 0^{(7)})$&1.0420&0.5311&11.82 & 60.93\\

  \cline{2-11}
    & &100&$(5, 0^{(9)})$&(0.4599, 2)&1.4459&$(5, 0^{(9)})$&0.5014&1.6064&9.99 & 518.16\\
    & 10&125&$(5, 0^{(9)})$&(0.5437, 4)&0.7412&$(5, 0^{(9)})$&0.6930&0.8535&13.16 & 216.89\\
     &&150&$(3, 0^{(3)}, 1, 0^{(4)}, 1)$&(0.5276, 6)&0.5183&$(0, 5, 0^{(8)})$&0.8685&0.5602&7.48 & 121.59\\
    \cline{2-11}
     &&100&$(3,0^{(11)})$&(0.4430, 2)&1.4994&$(3, 0^{(11)}$&0.4777&1.6477&9.00 & 613.32\\
    & 12&125&$(3,0^{(11)})$&(0.5583, 4)&0.7955&$(3, 0^{(11)}$&0.6483&0.8862&10.23 & 278.45\\
     &&150&$(3,0^{(11)})$&(0.5933, 6)&0.5221&$(3, 0^{(11)}$&0.8144&0.5841&10.61 & 148.38\\
     \hline
 \multirow{9}{*}{\begin{turn}{90}
   Entropy
\end{turn}}   &    & 100&$(3, 5,0^{(5)})$&(0.4890, 2)&2.8809&$(8,0^{(6)})$&0.5506&3.0391&5.21 & 153.08\\
   & 7 &125&$(6, 0^{(2)},1,0^{(2)},1)$&(0.1003,4)&-1.5165&$(6,0^{(5)},2)$&0.7217&-0.7487&102.55 & 72.06\\
 &  &150&$(3,0^{(4)},1,4)$&(0.1002,6)&-3.8779&$(4,0^{(5)},4)$&0.8849&-3.5865&8.12 & 28.54\\
    \cline{2-11}
  &     & 100&$(1^{(4)},0^{(5)},1)$&(0.4676, 0)&24.1534&$(1^{(4)},0^{(5)},1)$&0.4676&24.1534&0.00 & 322.89\\
   & 10 &125&$(1^{(4)},0^{(5)},1)$&(0.5679, 4)&2.0713&$(1^{(4)},0^{(5)},1)$&0.6430&2.5499&18.77 & 119.11\\
  & &150&$(3,0^{(8)},2)$&(0.1002,6)&-4.1412&$(3,0^{(8)},2)$&0.8148&-3.8208&8.39 & 61.79\\
   \cline{2-11}
  &   &100&$(1^{(3)},0^{(9)})$&(0.4665, 0)&32.9902&$(1^{(3)},0^{(9)})$&0.4665&32.9902&0.00& 315.85\\
  & 12  &125&$(1^{(3)},0^{(9)})$&(0.6362, 0)&5.0285&$(1^{(3)},0^{(9)})$&0.6362&5.0285&0.00 & 132.90\\
  &   &150&$(2,0^{(10)},1)$&(0.7778, 4)&-3.3032&$(1^{(2)},0^{(9)},1)$&0.7830&-3.1161&6.00 & 78.39\\
     \hline
 
    \end{tabular}
    \label{m}
\end{table}
Tables~\ref{alpha=0.5}--\ref{m} present the constrained optimal designs obtained under different budget costs, sample sizes, effective sample sizes, and Weibull distributions. As expected, the available experimental budget has a significant influence on the optimal design. When the allowable budget cost \(C_B\) is small, the optimization favours shorter censoring times and fewer guaranteed failures to satisfy the budget constraint. As \(C_B\) increases, the feasible region expands, allowing the design to approach the unconstrained optimum. Consequently, the A-optimality criterion decreases and REL becomes progressively smaller.

The comparison between GPHCS-I and PHCS-I demonstrates the advantage of introducing the guaranteed number of failures. Under almost all parameter configurations, GPHCS-I produces smaller values of the objective function than PHCS-I while satisfying the same budget constraint. The corresponding RRS is consistently positive, confirming that GPHCS-I provides superior estimation efficiency. The improvement is particularly noticeable when the available budget is limited, indicating that the additional flexibility offered by the guaranteed number of failures is most beneficial under severe economic constraints.

The Weibull shape parameter also influences the magnitude of the improvement. For decreasing hazard rates \((\alpha=0.5)\), the improvement over PHCS-I is relatively modest because failures occur later during the experiment. In contrast, for increasing hazard rates \((\alpha=2)\), failures occur earlier, enabling GPHCS-I to utilize the available testing time more efficiently and resulting in larger gains in estimation precision. Similar conclusions are observed when the sample size or the effective sample size increases, although the marginal improvement gradually decreases because the Fisher information increases at a diminishing rate.

The entropy-optimal designs exhibit a different pattern from the A-optimal designs. The optimal censoring schemes depend more strongly on the underlying lifetime distribution and the available budget, reflecting the fundamentally different objective of maximizing the information content of the experiment. Nevertheless, GPHCS-I consistently produces entropy values that are at least as good as, and often better than, those obtained under PHCS-I, demonstrating the effectiveness of the proposed generalized censoring scheme from an information-theoretic perspective.
\subsection{Multi-objective Design}\label{multi}

To illustrate the performance of the proposed multi-objective optimization approaches, extensive numerical studies are carried out under different model and cost settings. In Table~\ref{minimax}, the fixed cost components are taken as $C_D=3$, $C_{\tau}=75$, and $C_0=50$. The lifetime distribution is assumed to follow a Weibull distribution with scale parameter $\eta=1$ and shape parameter $\alpha=0.5$, $1$, and $2$. The sample size is chosen as $n=10$, $15$, and $20$, while the maximum number of observed failures is fixed at $m=5$.

To investigate the effect of the cost structure, Table~\ref{Ct} reports the optimal compromise designs for different combinations of the failure cost and operating cost. In this case, the Weibull parameters are fixed at $\alpha=0.5$ and $\eta=1$, with $n=15$, $m=7$, and $C_0=50$. The failure cost is taken as $C_D=3$, $7$, and $10$, whereas the operating cost is chosen as $C_{\tau}=25$, $50$, and $75$.

In Tables~\ref{minimax} and \ref{Ct}, $RRS_1$ and $RRS_2$ measure the relative deviation of the compromise design from the individually optimal A-optimal and cost-optimal designs, respectively. They are defined as
\[
RRS_1=\left(\frac{\phi_1(\boldsymbol{\zeta})}{\phi_1(\boldsymbol{\zeta}^*_1)}-1\right)\times100\%,\qquad RRS_2=\left(\frac{\phi_2(\boldsymbol{\zeta}^*)}
{\phi_2(\boldsymbol{\zeta}^*_2)}-1
\right)\times100\%,
\]
where $\boldsymbol{\zeta}^*$ denotes the corresponding compromise design. Smaller values of $RRS_1$ and $RRS_2$ indicate that the compromise design is closer to the respective individually optimal design. The values of $\phi_1(\boldsymbol{\zeta}_1^*)$ and $\phi_2(\boldsymbol{\zeta}_2^*)$ are obtained from the single-objective optimization procedures described in Section~\ref{uncons}.

\begin{table}[hbt!]
    \centering
    \caption{Multi-objective design for different values of $\alpha$ and $n$}
 \begin{tabular}{|c|cc|cccc|cc|c|}
    \hline
   &$\alpha$&$n$&$\boldsymbol{R}^*$&$(T^*,l^*)$&$\bar\phi_2(\boldsymbol{\zeta}^*)$&$\bar\phi_1(\boldsymbol{\zeta^*})$&$RRS_1$&$RRS_2$&$\phi_1(\boldsymbol{\zeta}_1^*)$\\
      \hline
\multirow{9}{*}{\begin{turn}{90}
  Minimax
\end{turn}}&   & 10 &$(0^{(4)}, 5)$&(0.2306, 3)&1.6443&100.8849&0.9806&0.9807&50.9341\\
 & 0.5&  15& $(4, 0^{(3)}, 6)$&(0.0726, 4)&1.6367&101.6629&0.9791&0.9791&51.3667\\
&&20& $(10, 0^{(3)},5)$&(0.0527, 4)&1.6347&102.4415&0.9784&0.9784&51.7793\\
 \cline{2-10}
  &     & 10 &$(4, 0^{(4)}, 1)$&(0.2953, 3)&0.6277&109.2434&1.1220&1.1220&51.4801\\
  & 1& 15&$(3, 1,0,6,0)$&(0.0101, 4)&0.5929&109.6843&1.1120&1.0041&21.9331\\
   & &20&$(9, 0, 1, 0, 5)$&(0.0101, 4)&0.5766&110.1984&1.1045&1.1033&52.3641\\
  \cline{2-10}
    & &   10 &$(3,2, 0^{(3)})$&(0.1002, 3)&0.8889&116.8777&1.0875&1.0822&55.9882\\
 &  2 &15&$(9,1,0^{(3)})$&(0.1002, 3)&0.7568&116.7901&1.0692&1.0602&56.4407\\
 & &  20& $(15, 0^{(4)})$&(0.1003, 3)&0.6552&117.7448&1.0703&0.9729&56.8711  \\
     \hline
\multirow{9}{*}{\begin{turn}{90}
  Compound
\end{turn}}  &     &10&$(0^{(3)},5,0)$&(0.0002,4)&1.4310&104.3895&1.0495&0.7237&50.9341\\
   & 0.5&15&$(5,0^{(2)},5,0)$&(0.0002,4)&1.5740&102.5707&0.9032&0.9968&51.3667\\
   & &20&$(10, 0^{(2)},5,0)$&(0.0002,4)&1.6532&101.7446&0.9650&1.0008&51.7793\\
    \cline{2-10}
    &&10&$(3,0^{(3)},2)$&(0.119, 3)&0.6889&103.3334&1.0072&1.3292&51.4801\\
   & 1&15&$(5,0^{(3)},5)$&(0.0101,4)&0.5463&112.2817&1.1620&0.9352&51.9331\\
    &&20&$(11,0^{(3)},4)$&(0.0101, 4)&0.5129&114.2262&1.1814&0.8710&52.3641\\
     \cline{2-10}
    & &   10 &$(5, 0^{(4)})$&(0.1002,3)&0.8217&119.4147&1.1328&0.9247&55.9882\\
  & 2 &15&$(10,0^{(4)})$&(0.1003,3)&0.7194&118.3191&1.0963&0.9582&56.4407\\
 & &  20& $(15, 0^{(4)})$&(0.1003, 3)&0.6552&117.7448&1.0703&0.9729&56.8711 \\
     \hline
  
    \end{tabular}
    \label{minimax}
\end{table}
Table~\ref{minimax} compares the minimax and compound formulations for the multi-objective optimization problem. The minimax criterion produces balanced compromise solutions by minimizing the largest normalized deviation among the competing objectives. Consequently, the two normalized objective values remain very close across all Weibull models and sample sizes, indicating that neither statistical efficiency nor experimental cost dominates the optimization process. This balanced behaviour is one of the principal advantages of the minimax formulation.

The compound criterion exhibits different characteristics. Since the optimization is based on a weighted aggregation of the objectives, the resulting designs are influenced by the relative contribution of each objective to the weighted sum. Consequently, one objective may improve substantially at the expense of the other, producing less balanced solutions than those obtained by the minimax criterion. Although both formulations generate feasible compromise designs, the minimax approach provides greater robustness when no prior preference between statistical efficiency and economic cost is available.
\begin{table}[hbt!]
    \centering
 \small   \caption{Multi-objective design for different values of $C_\tau$ and $C_D$ when $n=15$ and $m=7$}
 \begin{tabular}{|c|cc|cccc|cc|c|}
    \hline
&$C_\tau$&$C_D$&$\boldsymbol{R}^*$&$(T^*,l^*)$&$\phi(\boldsymbol{\zeta}^*)$&$TC(\boldsymbol{\zeta^*})$&$RRS_1$&$RRS_2$&$\phi_1(\boldsymbol{\zeta}_1^*)$\\
      \hline
 \multirow{9}{*}{\begin{turn}{90}
 Minimax
\end{turn}}&      25&3&$(0^{(6)},8)$&(0.2669, 6)&0.9024&77.0738&0.5288&0.5288&50.4148\\
      &25&7&$(4, 0^{5)},4)$&(0.1001, 5)&1.0967&94.6912&0.8580&0.8580&50.9649\\
      &25&10&$(6, 0^{(5)},2)$&(0.2977, 4)&1.2187&106.0760&1.0646&1.0646&51.3775\\
\cline{2-10}
      &50&3&$(0^{(6)},8)$&(0.1213, 6)&0.9687&82.7385&0.6411&0.6411&50.4171\\
      &50&7&$(3, 0^{(5)}, 5)$&(0.0812, 5)&1.1640&100.5003&0.9718&0.9718&50.9672\\
     & 50&10&$(1^{(2)},0^{(2)},6,0^{(2)})$& (0.0001, 5)&1.2716&111.9770&1.1793&1.1541&51.3797\\
\cline{2-10}
      &75&3&$(0^{(6)},8)$&(0.2769, 5)&1.0300&87.9764&0.7449&0.7449&50.4194\\
     & 75&7&$(0^{(6)},8)$&(0.1523,5)&1.2139&104.8189&1.0565&1.0564&50.9695\\
    &  75&10&$(1,0^{(3)},7,0^{(2)})$&(0.0001, 5)&1.3169&115.7616&1.2529&1.2310&51.3820\\
      \hline
 \multirow{9}{*}{\begin{turn}{90}
  Compound
\end{turn}}    &   25&3&$(0^{(6)},8)$&(0.6561, 6)&0.7941&79.8929&0.5847&0.3453&50.414\\
    &  25&7&$(2,0^{(4)},6,0)$&(0.0003, 6)&0.8953&101.7644&0.9967&0.5167&50.9649\\
    &  25&10&$(7,0^{(2)},1,0^{(3)})$&(0.0001,4)&1.3636&100.6156&0.9584&1.3101&51.3775\\
 \cline{2-10}
    &  50&3&$(0^{(5)},8,0)$&(0.0005, 6)&0.9899&82.0024&0.6265&0.6770&50.4171\\
     & 50&7&$(2,0^{(3)},2,3,1)$&(0.0001, 5)&1.2484&97.3161&0.9094&1.1149&50.9672\\
     & 50&10&$(3,0^{(3)},5,0^{(2)})$&(0.0001, 5)&1.1801&114.7203&1.2328&0.9991&51.3797\\
  \cline{2-10}
     & 75&3&$(0^{(5)},8,0)$&(0.0005, 6)&0.9899&89.0036&0.7653&0.6770&20.4194\\
     & 75&7&$(2,0^{(3)},4,0,2)$&(0.0001, 5)&1.2484&103.4741&1.0301&1.1149&50.9695\\
     & 75&10&$(1,0^{(3)},6,0,1)$&(0.0001, 5)&1.3169&115.7616&1.2529&1.2310&51.3820\\
      \hline
    \end{tabular}
    \label{Ct}
\end{table}
Table~\ref{Ct} investigates the influence of the inspection cost \(C_{\tau}\) and the failure cost \(C_D\) on the multi-objective designs. As expected, increasing either cost component leads to a higher total experimental cost and therefore encourages shorter experiments with more economical censoring schemes. The optimal progressive censoring pattern adapts automatically to these changes, demonstrating the flexibility of the proposed optimization framework.

The minimax solutions remain remarkably stable over different cost structures, with the normalized objective values remaining close to one another. This indicates that the minimax criterion successfully maintains a balanced compromise despite changes in the economic environment. In contrast, the compound formulation exhibits greater sensitivity to the cost parameters because the weighted objective function changes directly with the relative magnitude of the cost and information components. The numerical investigation confirms that the proposed optimization methodology effectively balances statistical efficiency and experimental cost while exploiting the additional flexibility of the GPHCS-I scheme. The theoretical monotonicity properties established in the paper also contribute to the computational efficiency of the optimization algorithm by substantially reducing the feasible search space.

\section{Data Analysis}\label{data}

Progressively censored data have been widely used in reliability studies since the pioneering work of \citet{herd1956estimation}, whose PhD thesis presented one of the earliest real-life applications involving gyroscope lifetime data. In this section, we illustrate the practical applicability of the proposed methodology using a real data set. The data, originally reported by \citet{Nelson_book} (Table~6.1, p.~228), consist of failure times of an insulating fluid obtained from an accelerated life test. The data were subsequently analyzed by \citet{viveros1994interval} in the context of progressive censoring and have since been widely used to assess statistical methods for censored life-testing experiments.

Recently, \citet{chakraborty2025application} analyzed this data set and showed that the Weibull distribution provides an adequate fit to the observed failure times. The estimated Weibull parameters are $\alpha=0.974$ and $\eta=0.108$. Following their analysis, we assume that the lifetimes follow a Weibull distribution with these parameter values. The original experiment consists of $n=18$ test units with a maximum of $m=8$ observed failures. These settings are used to illustrate the proposed optimal design methodology under the GPHCS-I scheme.

The cost coefficients are taken to be the same as those used in Section~\ref{numerical}, namely, $C_0=50$, $C_D=10$, and $C_{\tau}=50$. For the constrained optimization problem, the available budget is fixed at $C_B=65$. The resulting optimal designs under different design criteria are reported in Table~\ref{real_cons}.
\begin{table}[hbt!]
    \centering
 \caption{Design for real data-set}
 \begin{tabular}{c|c|ccccc}
    \hline
Condition&Design&$\boldsymbol{R}^*$&$(T^*,l^*)$&A-optimal Value&Total Cost&Entropy\\
 \hline
\multirow{2}{*}{Unconstrained} & A-optimal&$(0,11,0^{6)})$&(44.4147, 7)&0.0508&1314.84&36.9473\\   
 & Entropy&$(0,3,7,1,0^{(4)})$&(0.0824, 3)&0.1177&159.779&3.8987\\
 &Cost&$(11, 0^{(7)})$&(0.0823, 0)&0.1232&55.9165&22.3619\\
      \hline
 \multirow{2}{*}{Constrained} & A-optimal&$(0,10,0^{(5)},1)$&(0.2082, 0)&0.0656&64.9981&51.4266\\   
&  Entropy&$(11,0^{(7)})$&(0.2153, 0)&0.0700&64.9922&20.9194\\
  \hline
  \multirow{2}{*}{Multi-objective}&Minimax&$(0,11,0^{(6)})$&(0.2503, 0)&0.0617&67.9520&20.5877\\
  &Compound&$(0,11,0^{(6)})$&(0.2565, 0)&0.0612&68.3809&34.1618\\
  \hline
  All-three&Minimax&\multirow{1}{*}{$(4, 5, 2,0^{(5)})$}&\multirow{1}{*}{(2.2998, 0)}&\multirow{1}{*}{0.05218}&\multirow{1}{*}{193.1561}&\multirow{1}{*}{10.59478}\\
  \hline
    \end{tabular}
    \label{real_cons}
\end{table}

Table~\ref{real_cons} demonstrates that the proposed optimization procedures lead to substantially different GPHCS-I designs depending on the design objective. Under unconstrained optimization, the A-optimal design recommends a relatively large censoring time together with a large guaranteed number of failures, yielding the smallest A-optimality value at the expense of a considerably higher experimental cost. In contrast, the entropy-optimal design recommends an earlier termination of the experiment, leading to a lower-cost design while preserving a high level of information in the observed sample. As expected, the cost-optimal design minimizes the experimental cost but results in inferior statistical efficiency compared with both the A-optimal and entropy-optimal designs.

When the budget constraint is imposed, both the A-optimal and entropy-optimal procedures automatically adjust the censoring time and the progressive removal pattern to satisfy the prescribed budget. The constrained A-optimal design provides the best estimation precision among all feasible designs, whereas the entropy-optimal design achieves the minimum entropy under the same budget. These results illustrate the flexibility of the proposed methodology in accommodating practical economic limitations without sacrificing statistical efficiency.

The multi-objective optimization results further demonstrate the effectiveness of simultaneously balancing statistical efficiency and experimental cost. For this data set, both the minimax and compound approaches identify the same progressive removal vector, although the corresponding censoring times differ slightly. The minimax solution yields a marginally smaller experimental cost, whereas the compound design provides a slightly better A-optimality value. Finally, when entropy is incorporated as a third objective, the proposed minimax formulation successfully determines a compromise design that simultaneously balances estimation precision, information content, and experimental cost. This illustrates the flexibility of the minimax framework in handling multiple conflicting objectives, a feature that is difficult to achieve using conventional compound optimal design methods.

\section{Conclusion}\label{con}

This paper developed optimal life-testing designs under  GPHCS-I by considering statistical efficiency, information content, and experimental cost. A new expression for the Shannon entropy under GPHCS-I was derived and employed as an information-theoretic design criterion alongside the A-optimality criterion. Theoretical monotonicity properties of the A-optimality criterion and the cost function were established and exploited to reduce the optimization search space, thereby improving computational efficiency.

Unconstrained, constrained, and multi-objective optimization models were developed, and an efficient variable neighborhood search algorithm was proposed to determine the optimal censoring schemes. The numerical results demonstrated that the proposed GPHCS-I designs consistently outperform the conventional progressive hybrid censoring scheme in terms of estimation efficiency and information content under comparable cost constraints. Furthermore, the minimax formulation provided balanced compromise solutions and offered greater scalability and computational advantages over the compound formulation for multi-objective optimization.

The proposed methodology is sufficiently general to accommodate other lifetime distributions and alternative optimality criteria. Future research may consider Bayesian optimal designs, adaptive progressive hybrid censoring schemes, robust optimization under parameter uncertainty, and Pareto-based evolutionary multi-objective optimization methods. Another promising direction is to incorporate Shannon entropy directly into the multi-objective optimization framework, thereby simultaneously optimizing statistical precision, information content, and experimental cost. Such extensions would provide a more comprehensive decision-support framework for designing efficient and economically feasible reliability experiments.

\bibliographystyle{apalike}
\bibliography{BaysWarranty}

\vspace{0.5cm}
\appendix
\section{Compound Design Plots}
\normalsize

The appendix presents all graphical illustrations corresponding to the compound optimal design obtained in the numerical studies. These figures illustrate the variation of the normalized objective functions with respect to the preference weight $\lambda$ and identify the cut-point $\lambda^{*}$, at which the optimal compound design is obtained.


\begin{figure}[htbp]
\centering

\begin{subfigure}{0.3\textwidth}
\centering
\includegraphics[width=\linewidth]{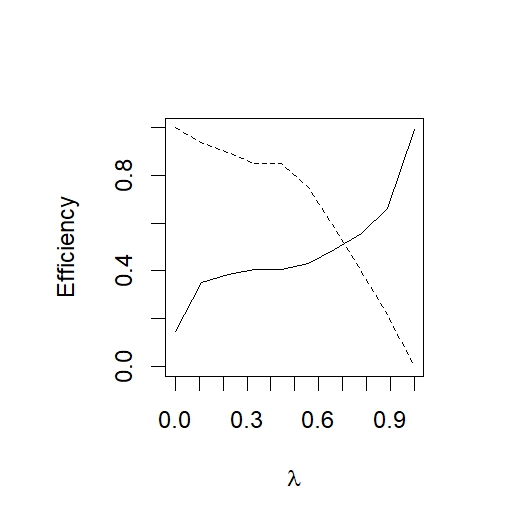}
\caption{$n=10,\;m=5$.}
\end{subfigure}
\hfill
\begin{subfigure}{0.3\textwidth}
\centering
\includegraphics[width=\linewidth]{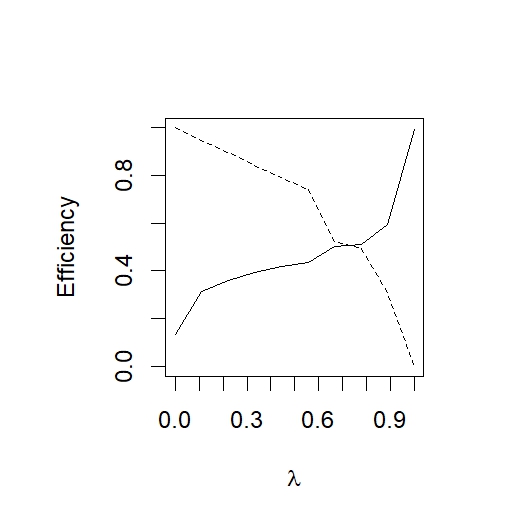}
\caption{$n=15,\;m=5$.}
\end{subfigure}
\begin{subfigure}{0.3\textwidth}
\centering
\includegraphics[width=\linewidth]{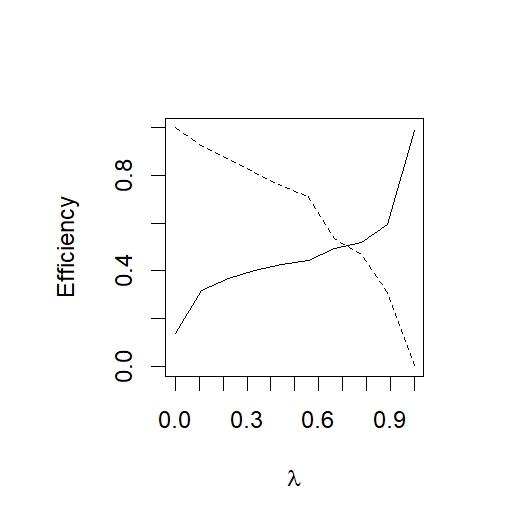}
\caption{$n=20,\;m=5$.}
\end{subfigure}
\caption{Compound objective function plots for the Weibull distribution with $(\alpha,\eta)=0.5,1$.}
\label{fig:appendix05}
\end{figure}
\begin{figure}[htbp]
\centering

\begin{subfigure}{0.3\textwidth}
\centering
\includegraphics[width=\linewidth]{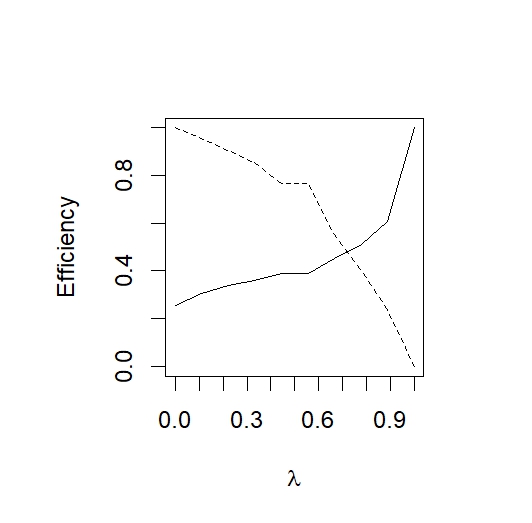}
\caption{$n=10,\;m=5$.}
\end{subfigure}
\begin{subfigure}{0.3\textwidth}
\centering
\includegraphics[width=\linewidth]{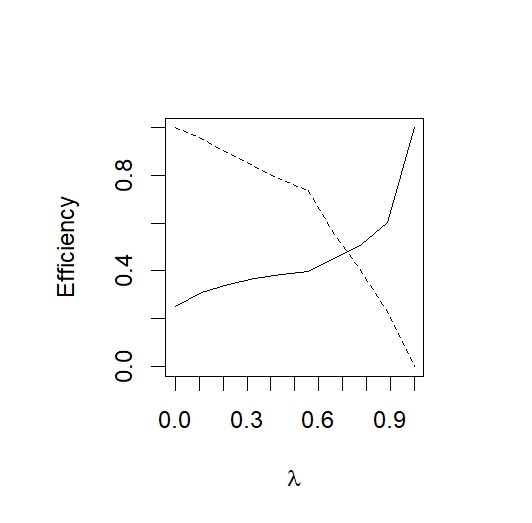}
\caption{$n=15,\;m=5$.}
\end{subfigure}
\begin{subfigure}{0.3\textwidth}
\centering
\includegraphics[width=\linewidth]{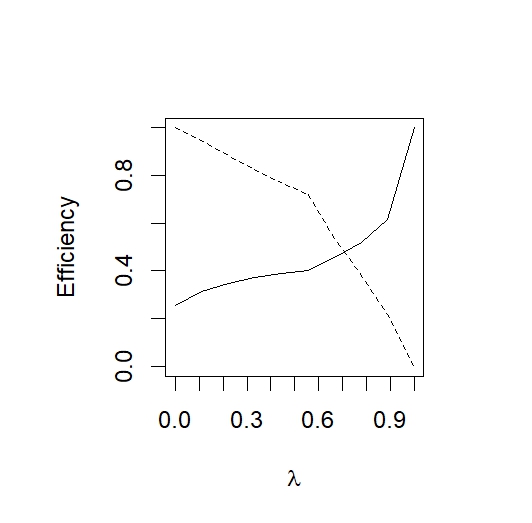}
\caption{$n=20,\;m=5$.}
\end{subfigure}
\caption{Compound objective function plots for the Weibull distribution with $(\alpha,\eta)=(1,1)$.}
\label{fig:appendix1}
\end{figure}

\begin{figure}[htbp]
\centering

\begin{subfigure}{0.3\textwidth}
\centering
\includegraphics[width=\linewidth]{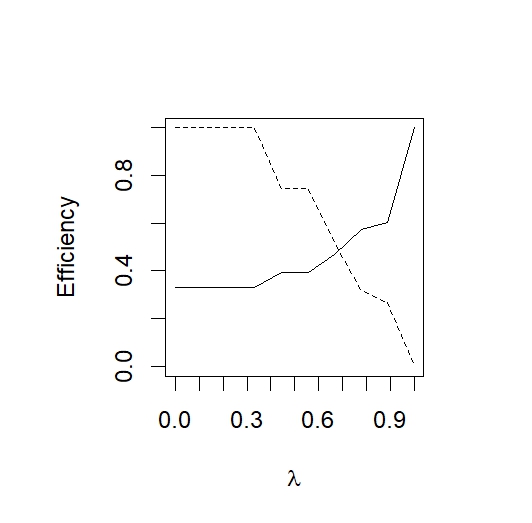}
\caption{$n=10,\;m=5$.}
\end{subfigure}
\begin{subfigure}{0.3\textwidth}
\centering
\includegraphics[width=\linewidth]{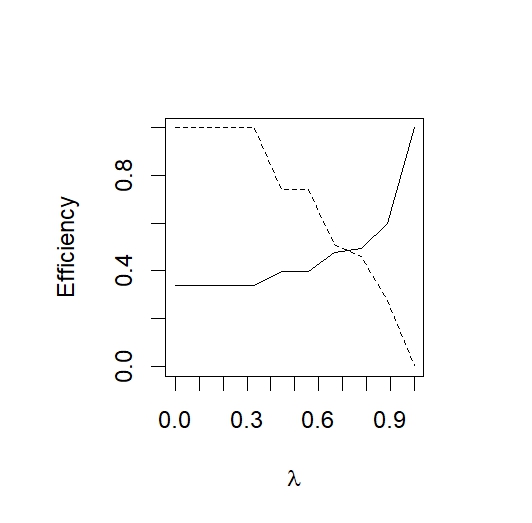}
\caption{$n=15,\;m=5$.}
\end{subfigure}
\begin{subfigure}{0.3\textwidth}
\centering
\includegraphics[width=\linewidth]{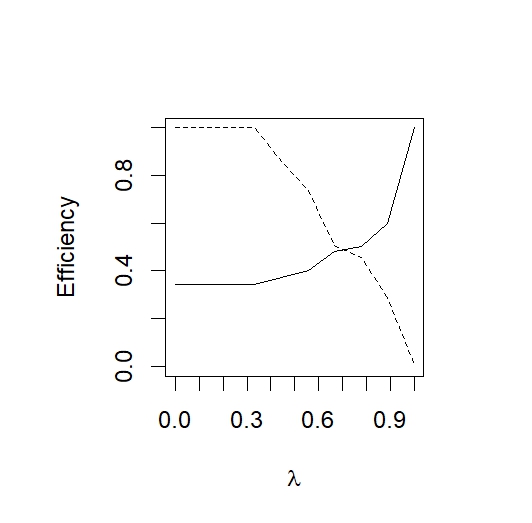}
\caption{$n=20,\;m=5$.}
\end{subfigure}
\caption{Compound objective function plots for the Weibull distribution with $(\alpha,\eta)=(2,1)$.}
\label{fig:appendix2}
\end{figure}

\begin{figure}[htbp]
\centering
\begin{subfigure}{0.3\textwidth}
\centering
\includegraphics[width=\linewidth]{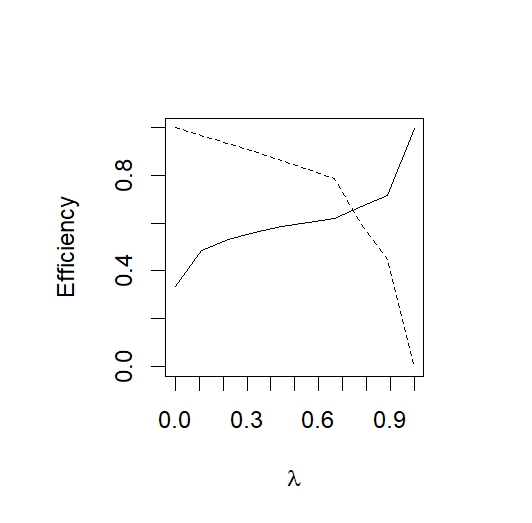}
\caption{$C_{\tau}=25,\;C_D=3$.}
\end{subfigure}
\begin{subfigure}{0.3\textwidth}
\centering
\includegraphics[width=\linewidth]{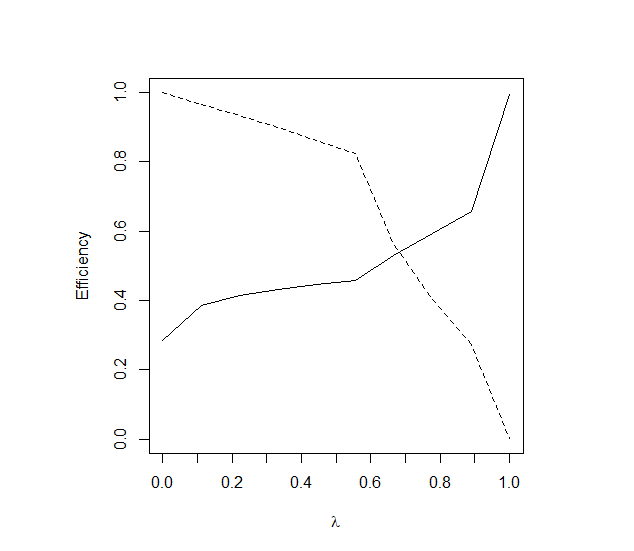}
\caption{$C_{\tau}=25,\;C_D=7$.}
\end{subfigure}
\begin{subfigure}{0.3\textwidth}
\centering
\includegraphics[width=\linewidth]{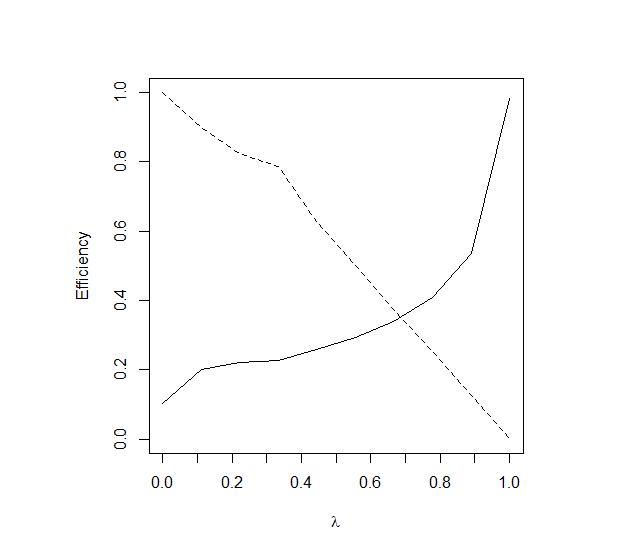}
\caption{$C_{\tau}=25,\;C_D=10$.}
\end{subfigure}
\begin{subfigure}{0.3\textwidth}
\centering
\includegraphics[width=\linewidth]{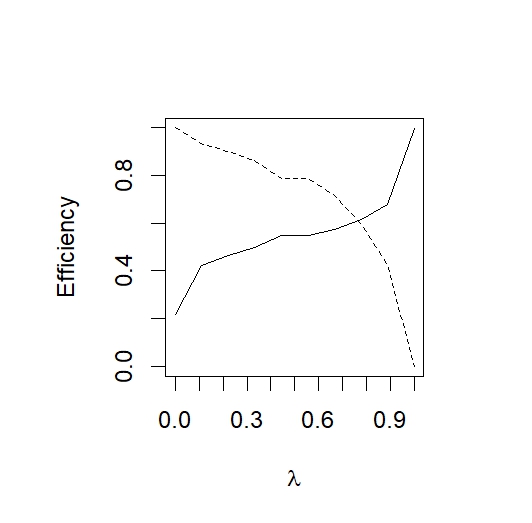}
\caption{$C_{\tau}=50,\;C_D=3$.}
\end{subfigure}
\begin{subfigure}{0.3\textwidth}
\centering
\includegraphics[width=\linewidth]{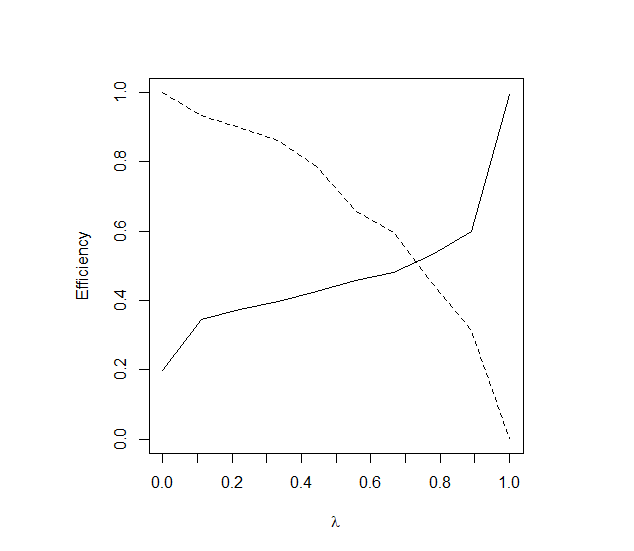}
\caption{$C_{\tau}=50,\;C_D=7$.}
\end{subfigure}
\begin{subfigure}{0.3\textwidth}
\centering
\includegraphics[width=\linewidth]{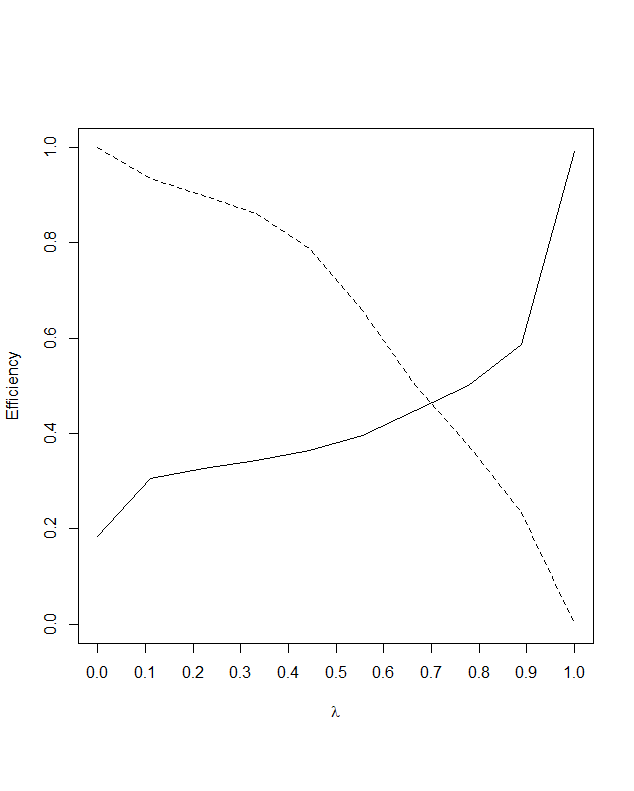}
\caption{$C_{\tau}=50,\;C_D=10$.}
\end{subfigure}
\begin{subfigure}{0.3\textwidth}
\centering
\includegraphics[width=\linewidth]{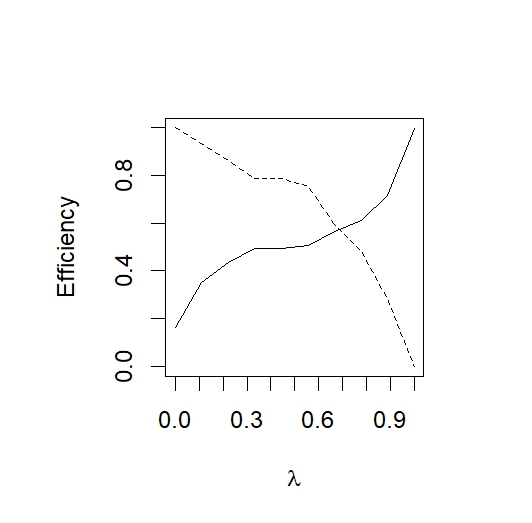}
\caption{$C_{\tau}=75,\;C_D=3$.}
\end{subfigure}
\begin{subfigure}{0.3\textwidth}
\centering
\includegraphics[width=\linewidth]{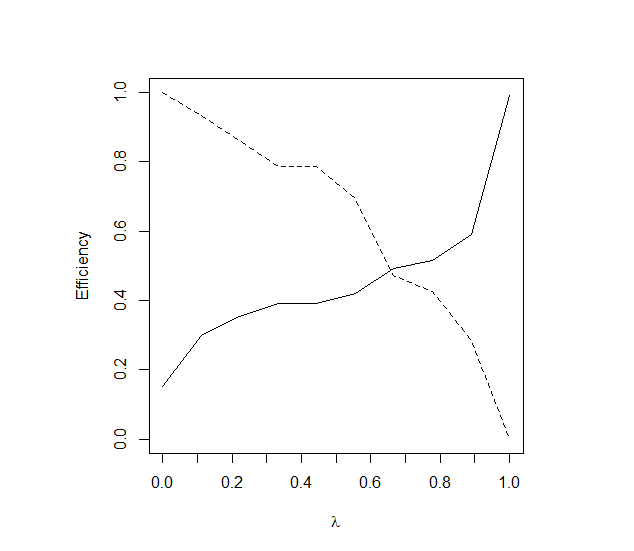}
\caption{$C_{\tau}=75,\;C_D=7$.}
\end{subfigure}
\begin{subfigure}{0.3\textwidth}
\centering
\includegraphics[width=\linewidth]{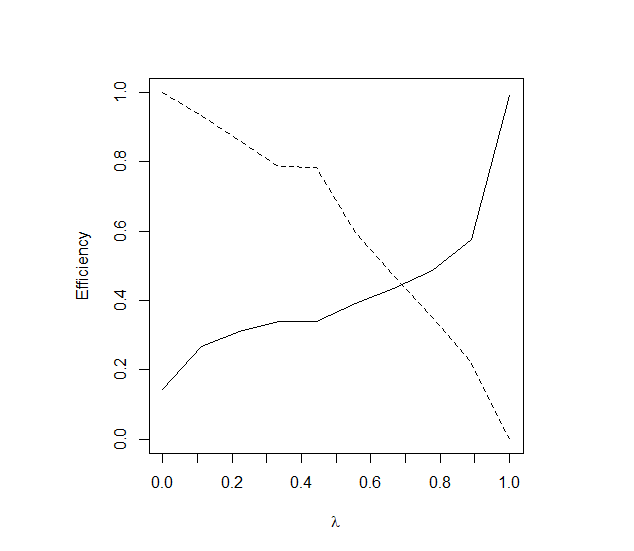}
\caption{$C_{\tau}=75,\;C_D=10$.}
\end{subfigure}
\caption{Compound objective function plots corresponding to the sensitivity analysis reported in Table~\ref{Ct}.}
\label{fig:appendixCt}
\end{figure}

\end{document}